\pdfoutput=1
\documentclass[A4]{lipics-v2021}
\usepackage{amstext,amsgen,latexsym,amsmath}
\usepackage{amssymb,amsfonts}
\usepackage{amsthm} 
\usepackage{pifont}
\usepackage{bm}
\usepackage{tikz}
\usepackage[ruled,linesnumbered,noend]{algorithm2e}
\DontPrintSemicolon
\usepackage{algpseudocode}
\usetikzlibrary{positioning,decorations.pathreplacing}
\usetikzlibrary{arrows.meta}
\usepackage{tikz-qtree}

\def\F{\mathbb{F}}
\newcommand{\SMP}{\mathsf{SMP}}
\newcommand{\set}[1]{ \left\{ #1 \right\} }

\newcommand{\abs}[1]{ \vert #1 \vert }

\newcommand{\MEQ}{\mathrm{MEQ}}
\newcommand{\MEQit}{\mathit{MEQ}}

\newcommand{\prob}[1]{\mathrm{#1}}
\newcommand{\ket}[1]{|{#1}\rangle}

\newcommand{\ignore}[1]{}

\hideLIPIcs
\nolinenumbers

\title{Quantum Simultaneous Protocols without Public Coins using Modified Equality Queries}


\author{Fran{\c c}ois Le Gall}
{Graduate School of Mathematics, Nagoya University, Japan}
{legall@math.nagoya-u.ac.jp}
{}
{
	JSPS KAKENHI grants Nos.~JP20H05966, 20H00579, 24H00071, MEXT Q-LEAP grant No.~JPMXS0120319794 and JST CREST grant No.~JPMJCR24I4.
}

\author{Oran Nadler}
	{Blavatnik School of Computer Science, Tel Aviv University, Israel}
	{oran.nadler@gmail.com}
{}
{}

\author{Harumichi Nishimura}
{Graduate School of Informatics, Nagoya University, Japan}
{hnishimura@i.nagoya-u.ac.jp}
{}
{
	JSPS KAKENHI grants Nos.~JP20H05966, 22H00522, 24H00071, 24K22293,  MEXT Q-LEAP grant No.~JPMXS0120319794 and JST CREST grant No.~JPMJCR24I4.
}

\author{Rotem Oshman}
	{Blavatnik School of Computer Science, Tel Aviv University, Israel}
	{roshman@tau.ac.il}
{}
{
ISF grants Nos.~2801/20 and 3725/24 and NSF-BSF grant No.~2022699.}

\authorrunning{F.\ Le Gall, O.\ Nadler, H.\ Nishimura and R.\ Oshman}


\Copyright{Fran{\c c}ois Le Gall, Oran Nadler, Harumichi Nishimura and Rotem Oshman}

\EventEditors{Silvia Bonomi, Letterio Galletta, Etienne Rivi\`{e}re, and Valerio Schiavoni}
\EventNoEds{4}
\EventLongTitle{28th International Conference on Principles of Distributed Systems (OPODIS 2024)}
\EventShortTitle{OPODIS 2024}
\EventAcronym{OPODIS}
\EventYear{2024}
\EventDate{December 11--13, 2024}
\EventLocation{Lucca, Italy}
\EventLogo{}
\SeriesVolume{324}
\ArticleNo{15}

\nolinenumbers

\begin{document}
\maketitle
\begin{abstract}
  In this paper we study a quantum version of the multiparty simultaneous message-passing ($\SMP$) model, and we show that in some cases, quantum communication can replace public randomness, even with no entanglement between the parties. This was already known for two players, but not for more than two players, and indeed, so far all that was known was a negative result. Our main technical contribution is a compiler that takes any classical public-coin simultaneous protocol based on ``modified equality queries,'' and converts it into a quantum simultaneous protocol without public coins with roughly the same communication complexity. We then use our compiler to derive protocols for several problems, including frequency moments, neighborhood diversity, enumeration of isolated cliques, and more. 
\end{abstract}

\section{Introduction}
\label{sec:intro}

In the \emph{multiparty simultaneous message-passing} ($\SMP$) {\emph{model} 
we have $k$ players with private inputs $x_1 ,\ldots,x_k \in \set{0,1}^n$,
and we would like to compute a function $f(x_1,\ldots,x_k)$ of the joint inputs.
To this end, each player $\ell \in [k]$ computes a message $M_\ell$,
which it sends to a \emph{referee}.
The referee collects all $k$ messages and produces an output, which should equal $f(x_1,\ldots,x_k)$,
except possibly with some small error probability.
Our goal is to use as little \emph{communication} as possible, that is,
the total number of bits sent by the players to the referee should be minimized.
There are several well-studied {classical} variants of the $\SMP$ model:
a \emph{deterministic} variant, where the participants (i.e., the players and the referee) have no randomness, and no error is allowed;
a \emph{public-coin} variant, where the participants have access to a common random string;
and a \emph{private-coin} variant, where each participant has access to its own random string. 
The public-coin variant is the strongest of the three, but arguably, it is unrealistic:
in the absence of prior coordination, many distributed systems
do not have a source of common randomness, and must make do 
with the private randomness available to each participant.%
\footnote{In \emph{non-simultaneous} protocols it is possible to replace public randomness with private randomness~\cite{Newman91}, 
but this requires at least one synchronized round of communication, and synchronization comes with its own costs.}

In this paper we study a \emph{quantum} version of the $\SMP$ model, and we show that in some cases, quantum communication can replace public randomness, even with no entanglement between the parties.
This was already known for two players \cite{BCWW01,Yao03}, but it was not known for more than two players, and indeed, so far all that was known was a negative result~\cite{Gavinsky+CCC13} showing that for $k = 3$ players, there exist problems for which quantum communication with no public randomness or entanglement is exponentially weaker than classical communication \emph{with} public randomness.
Our main technical contribution is a compiler that takes any classical public-coin simultaneous protocol based on \emph{modified equality queries} (which we define below), and converts it into a quantum \emph{private-coin} simultaneous protocol with roughly the same communication complexity. We then use our compiler to derive protocols for several problems, including \emph{frequency moments},  \emph{neighborhood diversity}~\cite{Lam12}, \emph{enumeration of isolated cliques}~\cite{KHMN09}, and more.\vspace{-2mm}

\subparagraph*{Modified equality queries.}
We observe that quite a few problems studied in the literature can be solved by repeatedly executing the following type of query,
$\MEQ_{k,n}(i, j, y, z)$: for
two indices $i,j \in [k]$, and two strings $y, z \in \set{0,1}^n$ known only to the referee,
is it the case that $x_i \oplus y = x_j \oplus z$? (Here, $x_i, x_j \in \set{0,1}^n$ are the inputs of players $i,j$, respectively.)
We refer to such queries as \emph{modified equality queries}.

There is a well-known quantum simultaneous protocol~\cite{BCWW01} for ``plain'' equality queries, where we merely wish to determine
whether $x_i = x_j$ (without the modifying strings $y,z$).
In the protocol of~\cite{BCWW01}, each player computes a short \emph{quantum fingerprint} of its input and sends it to the referee,
who then uses the quantum fingerprints to compare the players' inputs (with some error probability; see Section~\ref{sec:prelim} for the details).
We observe that quantum fingerprints also allow us to implement \emph{modified} equality queries,
and moreover, this can be done even if the players do not know the strings $y,z$:
\begin{lemma}
	For any $s\geq 1$, let $\mathcal{Q}_{s}$ denote the set of all $s$-qubit quantum states.%
	\footnote{An \emph{$s$-qubit quantum state} is any quantum state that can be represented in $s$
		qubits; it is essentially a quantum superposition over classical $s$-bit strings.
See Section~\ref{sec:prelim} for the precise definition.}
	For any $n \geq 1$ and $\varepsilon \in (0,1)$,
	there is a quantum operator%
	\footnote{A \emph{quantum operator} is an operation that maps one quantum state into another.
		In our case, we apply it to a classical string, which is slight notation abuse.}
	$F : \set{0,1}^n \rightarrow \mathcal{Q}_{s}$ with 
	$s=O(\log n \cdot \log(1/\varepsilon))$ 
	such that if each player $\ell$ sends $F(x_\ell)$ to the referee,
	then for any $i,j \in [k]$ and any $y,z \in \set{0,1}^n$,
	the referee can compute 
	$\MEQ_{k,n}(i, j, y, z)$
	with error probability at most $\varepsilon$.
	\label{lemma:MEQ}
\end{lemma}

We stress that Lemma \ref{lemma:MEQ} only states that the referee can compute $\MEQ_{k,n}(i, j, y, z)$ for \emph{one} 4-tuple $(i,j,y,z)$. Unlike classical protocols, in the quantum world it is not technically immediate to re-use information sent by the players to compute the value of the query for more than one 4-tuple (showing how to bypass this difficulty is indeed one of the main contributions of this paper).\vspace{-2mm}

\subparagraph*{Compiling $\boldsymbol{\MEQ}$ decision trees into quantum protocols.}
Although Lemma~\ref{lemma:MEQ} allows us to implement a single modified equality query,
by itself it is not enough to obtain an efficient protocol for many of the problems we want to solve,
as these problems require us to execute \emph{many} such queries.
For example,
in the \emph{distinct elements} problem, the goal is to determine the number of distinct values
among $x_1,\ldots,x_k$.
Solving this problem requires $\binom{k}{2}$ ``plain'' equality queries, as every player's input must be compared against all the others.
In general, our goal is to work with protocols represented by an \emph{$\MEQit_{k,n}$ decision tree}:
a rooted binary tree whose inner nodes are labeled by $\MEQ_{k,n}$ queries,
and whose leaves are labeled by output values (e.g., 0 or 1 if the tree computes a Boolean function).
The tree is evaluated starting from the root,
and at each step we evaluate the query written in the current node,
proceeding to the left child if the answer is~0 and to right child if the answer is~1,
until we eventually reach a leaf and output the value written in it.

A na\"ive application of Lemma~\ref{lemma:MEQ} results in a quantum protocol 
whose communication cost scales {linearly} with the depth of the decision tree, as we must call the protocol from Lemma~\ref{lemma:MEQ} at each step.
However, we can do much better:
we show that we can compile a decision tree into a quantum protocol whose communication cost
depends only {logarithmically} on the depth of the tree.
The key is to \emph{re-use information}:
instead of evaluating each modified equality query on its own,
we would like to re-use the information sent by the players, so that evaluating multiple queries that involve the same player $i$
will not require player $i$ to send fresh information each time.

As already mentioned, unlike classical protocols,
in the quantum world it is not technically immediate to re-use information sent by the players. First, quantum states cannot be duplicated (this is a consequence of the no-cloning theorem in quantum information theory).
Second, if at any point we \emph{measure} a quantum state, we may cause it to collapse,
losing all the information that was stored in it (except for the outcome of the measurement),
and preventing it from being re-used.
This indeed happens in the protocol from Lemma~\ref{lemma:MEQ}.
To avoid this pitfall, we use \emph{gentle measurements} (see, e.g., \cite[Section 1.3]{Aar2016}),
relying on the fact that a measurement whose outcome is ``nearly certain''
has very little effect on the quantum state we are measuring.
To ensure that the outcome of each measurement we make is ``nearly certain'', 
we \emph{amplify} the success probability of each query $\MEQ_{k,n}(i, j, y, z)$ so that
if $x_i \oplus y = x_j \oplus z$ then the measurement returns 1 with probability nearly 1,
and if $x_i \oplus y \neq x_j \oplus z$ then the measurement returns 1 with probability nearly 0. 
Finally, we observe that the quantum union bound by Gao \cite{Gao15} can be used
and conclude that the measurements can be applied sequentially on the same state with only a small decrease of the success probability.

Ultimately, our result is the following:
\begin{theorem}\label{th:main}
	For any $n, k, D \geq 0$ and $\delta \in (0,1)$,
	any $\MEQ_{k,n}$ decision tree of depth $D$ can be implemented
	by a quantum $k$-party $\SMP$ protocol that uses
	$O(k  (\log D + \log(1/\delta))\log n)$ qubits
	and has error probability at most $\delta$.
\end{theorem}

\subparagraph*{Applications.}
We give several applications of our compiler in Section \ref{sec:app}.
Several are technically straightforward. For instance, using ``plain'' equality queries, we can compare all the players' inputs to one another, which allows us to count the number of distinct elements or compute other frequency moments of the input.
Next we turn to more complex applications involving graphs:
we show that in the number-in-hand network model~\cite{BMRT14} (a special case of the $\SMP$ model also sometimes called \emph{broadcast congested clique}),
we can use our compiler to obtain efficient simultaneous quantum protocols for $P_3$- and $P_4$-induced subgraph freeness \cite{KMRS15,MPRT20},
computing neighborhood diversity~\cite{Lam12},
enumerating isolated cliques~\cite{KHMN09} and
reconstructing distance-hereditary graphs~\cite{KMRS15,MPRT20}.
For all these problems, we obtain efficient quantum protocols that do not require public randomness:
in all of our protocols, each player only sends $\mathrm{polylog}(n,k)$ qubits. This cost matches the cost of public-coin classical protocols and improves exponentially the cost of private-coin classical protocols.\footnote{For private-coin classical protocols a lower bound of the form $\Omega(\sqrt{n})$ or $\Omega(\sqrt{k})$ trivially follows from the lower bound on the cost of private-coin classical protocols for the two-party equality function \cite{Babai+97,Newman1996}.}\vspace{-3mm}

\subparagraph*{Relation with prior works on quantum distributed computing.}
Several works \cite{A+24,AVPODC22,Arfaoui+14,CR+STOC24,Denchev+08,ElkinKNP14,FLNP21,Fraigniaud2024,Gavinsky+CCC13,GavoilleKM09,Hasegawa+PODC24,Izumi+PODC19,Izumi+STACS20,LeGall+2023,LeGall+PODC18,LMN23,LeGall+STACS19,Tani+12,WYPODC22} have investigated how quantum communication can help for various computational tasks and settings in distributed computing. To our knowledge, theoretical aspects of the quantum multi-party $\SMP$ model have only been considered in Ref.~\cite{Gavinsky+CCC13}, which we already mentioned, and Ref.~\cite{Gavinsky+2008}, which focuses on a different input model (the number-on-the-forehead model). 
There are also a few experimental investigations of multiparty simultaneous quantum protocols \cite{GS20,Q+21}, but these works are mainly empirical and not concerned with asymptotic complexity.

\section{Preliminaries}
\label{sec:prelim}

\subparagraph*{Notation and terminology.}
For any integer $n\ge 1$, we write $[n]=\{1,\ldots,n\}$. 
For any strings $x,x'\in\{0,1\}^n$, we denote by $\Delta_n(x,x')$ the Hamming distance between $x$ and $x'$. 

In this paper we consider undirected graphs with no self-loops $G = (V,E)$ over $k = |V|$ nodes (since the number of nodes will always match the number of players in the protocol, we use the same notation $k$ for both).
We use $\deg(v)$ to denote the degree of node $v\in V$.
We often implicitly assume that $V = \set{1,\ldots,k}$.
Let $N(v) \subseteq V$ denote the neighbors of node $v \in V$ in~$G$,
and $\nu_{v} \in \set{0,1}^k$ denote the characteristic vector 
of $N(v)$, where $\nu_{v}[u] = 1$ iff $\set{ v, u } \in E$ for each $u \in [k]$.
Let $e_{v} \in \set{0,1}^k$ be the 
characteristic vector of the singleton $\set{v}$, 
i.e., the vector where $e_{v}[u] = 1$ iff $u = v$.
For a subset $S \subseteq V$, we use $G[S]$ to denote the subgraph of~$G$ induced by $S$. For a node $v\in V$, we define $G-v$ as $G-v=G[V\backslash\{v\}]$.

A node $v$ in $G$ is called {\em pendant} if $v$ has only one neighbor. 
Two nodes $u,v$ in $G$ are called {\em false twins} (resp., {\em true twins}) if $u$ and $v$ are not adjacent (resp., adjacent), 
and have the same neighborhood, that is, $N(u)=N(v)$ (resp., $N(u) \cup \set{u} = N(v) \cup \set{v}$,
or equivalently,
	$N(u)\setminus \{v\}=N(v)\setminus\{u\}$).
We say that  $u,v$ are \emph{twins} if they are either true twins or false twins.
Note that
two non-adjacent nodes $u,v$ cannot have $N(u) \cup \set{u} = N(v) \cup \set{v}$,
and because there are no self-loops in the graph,
two adjacent nodes $u,v$ cannot have $N(u) = N(v)$.
Therefore, in terms of neighborhood vectors, we have:
\begin{proposition}
	Nodes $u\neq v$ are false twins if and only if $\nu_u = \nu_v$,
	and true twins if and only if $\nu_u \oplus e_u = \nu_v \oplus e_v$.
	\label{prop:twins}
\end{proposition}


\subparagraph*{SMP protocols and NIH network model.}
A \emph{simultaneous message-passing} ($\SMP$) \emph{protocol}
features $k$ players with inputs $x_1,\ldots,x_k \in \set{0,1}^n$, respectively,
and a referee, who does not know $x_1,\ldots,x_k$.
In the protocol, each player sends one message to the referee,
and the referee then produces an output.
The goal of the referee is to compute some function $f(x_1,\ldots,x_k)$ of the inputs,
and we say that the protocol \emph{succeeds} whenever the referee's output is correct.
The \emph{communication cost} of the protocol is the maximum total number of bits sent 
by the players to the referee in any execution of the protocol, on any input. 
We say that a protocol is bounded-error if for any input, it outputs the correct answer with probability at least $2/3$. 
In this paper we do not assume that the players have shared randomness;
each player's message depends only on its own input. 

A special case of the $\SMP$ model is the \emph{number-in-hand (NIH) network model}~\cite{BMRT14}.
Here, the input to the computation is an undirected graph $G = (V,E)$ over $k$ nodes,
and each party represents a node in the graph.
The input to player $v \in [k]$ is the neighborhood vector $\nu_v \in \set{0,1}^k$ (we thus have $n = k$ in this case),
and the referee is asked to solve some graph problem on $G$.

In the quantum versions of the $\SMP$ model and the NIH network model, the only difference is that players are allowed to send quantum information to the referee. The \emph{communication cost} of the protocol is the maximum total number of quantum bits (qubits) sent by the players to the referee in any execution of the protocol. We do not assume that the players have shared randomness or shared entanglement; each player's message depends again only on its own input.\vspace{-3mm}

\subparagraph*{Basics of quantum information.}
The most basic notion in quantum information is the concept of quantum bit (qubit), which represents the state of an elementary physical system that follows the laws of quantum mechanics (e.g., one photon). Qubits are physically stored in \emph{quantum registers}. 
Mathematically, the state of a quantum register consisting of $q$ qubits is described by a unit-norm complex vector of dimension $m$, where $m=2^q$, and usually written using Dirac's notation as $\ket{\psi}$. By taking an orthonormal basis of the $m$-dimensional complex vector space and indexing these basis vectors (again using Dirac's notation) as $\ket{j}$ for all $j\in \{1,\ldots,m\}$, we can write
$
\ket{\psi}=\sum_{j=1}^{m} \alpha_j\ket{j}
$
for complex numbers $\alpha_j$ such that the state has norm 1 (i.e., satisfying $\sum_{j=1}^m \vert \alpha_j\vert^2=1$).

All transformations on quantum registers need to be unitary, i.e., described by unitary matrices. The main unitary matrices that will appear in the technical parts of this paper are the Hadamard gate (denoted $H$) and the Pauli $X$ gate, both acting on 1 qubit, and the $\mathit{CNOT}$ gate acting on two qubits. The precise definition of these gates will not be necessary for understanding this paper.

Information can only be extracted from a quantum register by measurements. The most elementary type of measurements is measurement of a 1-qubit register in the computational basis. For a 1-qubit register in the state 
$\ket{\psi}=\alpha\ket{0}+\beta\ket{1}$ where $\abs{\alpha}^2+\abs{\beta}^2=1$, measuring it in the computational basis gives as outcome 1 bit: the outcome is $0$ with probability $\abs{\alpha}^2$ and $1$ with probability $\abs{\beta}^2$. Importantly, the state collapses (i.e., is irreversibly modified) after the measurement: in the former case the postmeasurement state is $\ket{0}$, while in the later case the postmeasurement state is $\ket{1}$.

Several more general kinds of measurements are allowed by quantum mechanics. In this paper, we will mostly use the \emph{2-outcome measurements} defined as follows. A 2-outcome measurement of a $q$-qubit register ${\sf{R}}$ is the following process:  introduce a new 1-qubit register $\mathsf{S}$ initialized to $\ket{0}$, apply a unitary transform $U$ on the whole system, and then measure Register $\mathsf{S}$ in the computational basis, which gives as outcome a bit $b\in\{0,1\}$. We refer to Figure~\ref{fig:GM} in Appendix~\ref{Appendix:fig} for an illustration of the process. We denote such a 2-outcome measurement by $\mathcal{M}$. We will also use the following notation: for any bit $b\in\{0,1\}$ and any $q$-qubit quantum state $\ket{\psi}$, we denote by ${\cal M}^{b}(\ket{\psi})$ the probability of obtaining outcome $b$ by $\mathcal{M}$ when the initial state in $\sf{R}$ is $\ket{\psi}$.  

\subparagraph*{Quantum union bound.} We now present the quantum union bound by Gao \cite{Gao15}. While this bound only applies to a special type of 2-outcome measurements called 2-outcome projective measurements (defined in Appendix \ref{Appendix:pm}), any 2-outcome measurement can actually be efficiently converted  into a 2-outcome projective measurement (the conversion is described in Appendix \ref{Appendix:pm}).

Consider several 2-outcome projective measurements ${\cal M}_1,\ldots,{\cal M}_N$ acting on the same $q$-qubit register. Consider what happens when performing these $N$ measurements \emph{sequentially}. 
Specifically, assume that the system is initially in state $\ket{\psi}$.
We first perform ${\cal M}_1$ on $\ket{\psi}$ 
and obtain a postmeasurement state $\ket{\psi_1}$.
Then we perform ${\cal M}_2$ on $\ket{\psi_1}$ 
and obtain the postmeasurement state $\ket{\psi_2}$. 
And so it carries on, with each measurement being performed on the state resulting 
from the previous measurement. 
After $N$ measurements, we obtain the state $\ket{\psi_N}$. 
For an arbitrary binary string $s\in \{0,1\}^N$, we would like to estimate the probability that the sequence of outcomes of this measurement process is $s$, i.e., the probability that for all $i\in\{1,\ldots,N\}$, the outcome of measurement ${\cal M}_i$ is the bit~$s_i$. The following theorem by Gao \cite{Gao15} shows that this probability is high when for each $i\in\{1,\ldots,N\}$ applying measurement ${\cal M}_i$ on the \emph{initial state} $\ket{\psi}$ gives outcome~$s_i$ with high probability.

\begin{theorem}[Quantum union bound~\cite{Gao15}]\label{thm:Gao15}
    For any string $s\in \{0,1\}^N$, the probability that the above sequential measurement process has outcome $s$ is at least 
    \[
    1-4\sum_{i=1}^N \Big(1-{\cal M}^{s_i}_i(\ket{\psi})\Big). 
    \]
\end{theorem}


\subparagraph*{The SWAP test.}
The SWAP test~\cite{BBD+97,BCWW01} is a quantum protocol that checks whether two quantum states $\ket{\psi_1}$ and $\ket{\psi_2}$ stored in two $q$-qubit registers ${\sf R}_1$ and ${\sf R}_2$, respectively, are close or not (i.e., estimates their inner product). For completeness we give a detailed description of the test in Appendix \ref{Appendix:SWAP} (this detailed description is not needed to understand the claims of this paper). 
The main property of the SWAP test is that the test outputs $1$ with probability $\frac{1}{2}+\frac{1}{2}|\langle\psi_1|\psi_2\rangle|^2$ (and outputs 0 with probability $\frac{1}{2}-\frac{1}{2}|\langle\psi_1|\psi_2\rangle|^2)$, where $\langle\psi_1|\psi_2\rangle$ denotes the inner product between $|\psi_1\rangle$ and $|\psi_2\rangle$.



The SWAP test is especially useful when combined with the notion of quantum fingerprints. We first give the definition of this concept. 
\begin{definition}\label{def:fingerprint}
A quantum fingerprint family for the set of $n$-bit strings is a family $\{ |h_x\rangle: x\in\{0,1\}^n \}$ such that the following conditions hold for each $x\in\{0,1\}^n$:
\begin{itemize}
    \item[1.]
    $\ket{h_x}$ is a $O(\log n)$-qubit quantum state; 
    \item[2.] 
    $|\langle h_x|h_{x'}\rangle|\leq \zeta$ holds for all $x'\in\{0,1\}^n\setminus\{x\}$, for some universal constant $\zeta\in (0,1/2]$. 
\end{itemize}
\end{definition}

For a quantum fingerprint family $\{ |h_x\rangle: x\in\{0,1\}^n \}$, the SWAP test on states $\ket{\psi_1}=\ket{h_x}$ and $\ket{\psi_2}=\ket{h_{x'}}$ outputs 1 with probability~1 if $x=x'$ (since $\langle h_x|h_{x}\rangle=1$) and outputs 1 with probability at most $\frac{1}{2}+\frac{\zeta^2}{2}$ if $x\neq x'$ (since $|\langle h_x|h_{x'}\rangle|\leq \zeta$). For later reference, we state this result in the following lemma.

 \begin{lemma}\label{lem:swap-test}
 When $\ket{h_x}$ and $\ket{h_{x'}\!}$ are given in ${\sf R}_1\!$ and ${\sf R}_2$, respectively, 
 the SWAP test outputs 1 with probability 1 if $x=x'$, and outputs 1 with probability at most $\frac{1}{2}+\frac{\zeta^2}{2}\le \frac{5}{8}$ if $x\neq x'$. 
 \end{lemma}

 Ref.~\cite{BCWW01} showed how to create quantum fingerprint families. We will actually need a special kind of quantum fingerprint families, also used in \cite{GKW04}, that satisfies the following additional property: for any known string $y\in\{0,1\}^n$, the fingerprint of $x$ can be converted to the fingerprint of $x\oplus y$ (by a unitary transformation depending on $y$) without knowing the fingerprint of $x$. We call a quantum fingerprint family satisfying this additional property a \emph{linear quantum fingerprint family}. Ref.~\cite{GKW04} showed how to construct a linear quantum fingerprint family, which we write $\{|\Psi_x\rangle: x\in\{0,1\}^n \}$. 

 For completeness, we briefly describe the construction from \cite{GKW04}, which is 
based on constant rate linear error-correcting codes (the details of the construction will not be needed to understand the results in this paper). Take a linear function $E:\{0,1\}^n\rightarrow \{0,1\}^m$ where $m=O(n)$ such that 
$\Delta_m(E(x),E(x'))=\Omega(m)$
for any distinct $x,x'\in\{0,1\}^n$.
The corresponding quantum fingerprint of $x$ is then defined as the $O(\log n)$-qubit quantum state 
\[
|\Psi_x\rangle=\frac{1}{\sqrt{m}}\sum_{j=1}^m (-1)^{E(x)_j} |j\rangle,
\]
where $E(x)_j$ denotes the $j$th bit of $E(x)$. It is easy to check that this family of states satisfies all the required conditions. In particular, for any $y\in\{0,1\}^n$, the state $|\Psi_x\rangle$ can be converted to $|\Psi_{x\oplus y}\rangle$ by a unitary transformation (depending on $y$) without knowing $|\Psi_x\rangle$. 


\section{ \boldmath{ Quantum $\SMP$ Protocols Based on $\MEQ$ Decision Trees}}\label{sec:meq}
In this section we prove the main technical results of this paper (Lemma \ref{lemma:MEQ} and Theorem \ref{th:main}).

\subsection{Implementing a Single Query: Proof of Lemma \ref{lemma:MEQ}}\label{sub:prooflemma}

We first give a brief sketch of the proof. We use the linear quantum fingerprint family $\{|\Psi_x\rangle: x\in\{0,1\}^n \}$ introduced at the end of Section \ref{sec:prelim}. Each player sends the fingerprint corresponding to its input (i.e., player $\ell$ sends the state $|\Psi_{x_\ell}\rangle$). The referee then implements the SWAP test on the states $|\Psi_{x_i\oplus y}\rangle$ and $|\Psi_{x_j\oplus z}\rangle$, which can be constructed from the messages of player $i$ and the player $j$ due to the linearity property of the quantum fingerprint family. From Lemma \ref{lem:swap-test}, we know that the success probability of the SWAP test is at least $5/8$. We amplify the success probability by applying $O(\log(1/\varepsilon))$ SWAP tests in parallel, which requires each player to actually send $O(\log(1/\varepsilon))$ copies of its quantum fingerprint. We now explain all the details of the proof.

\begin{proof}[Proof of Lemma \ref{lemma:MEQ}]
Take $t=\Theta(\log(1/\varepsilon))$. For any $x\in\{0,1\}^n$ we define
$
F(x)=\ket{\Psi_{x}}^{\otimes t},
$
i.e., $t$ copies of the (linear) quantum fingerprint of $x$.
Since each state $\ket{\Psi_{x}}$ is encoded by $O(\log n)$ qubits, $F(x)$ is a quantum state of $O(t\log n)=O(\log(1/\varepsilon)\log n)$ qubits, as claimed. We now describe and analyze the referee's procedure.

\subparagraph*{Description of the referee's procedure.} 
Remember that the referee knows the indices $i,j$ and the strings $y,z$. The referee receives the quantum message $F(x_\ell)$ from player $\ell$, for each $\ell\in\{1,\ldots,k\}$. We assume that $F(x_\ell)$ is stored by the referee in registers $({\sf R}_{\ell,1},\ldots,{\sf R}_{\ell,t})$, where each ${\sf R}_{\ell,r}$ stores one copy of $\ket{\Psi_{x_\ell}}$.

\begin{figure}[ht!]
\begin{center}
  \fbox{
   \begin{minipage}{13 cm} \vspace{2mm}
  
  

  \noindent\hspace{3mm} 
  1. Convert $\ket{\Psi_{x_i}}^{\otimes t}$ into $\ket{\Psi_{x_i\oplus y}}^{\otimes t}$ in Registers $({\sf R}_{i,1},\ldots,{\sf R}_{i,t})$.

  \noindent\hspace{7mm} 
  Convert $\ket{\Psi_{x_j}}^{\otimes t}$ into $\ket{\Psi_{x_j\oplus z}}^{\otimes t}$ in Registers $({\sf R}_{j,1},\ldots,{\sf R}_{j,t})$.
  \vspace{2mm}

  \noindent\hspace{3mm}
  2. For every $r=1,\ldots,t$:
  
  \noindent\hspace{6mm}
  2.1. Introduce a register ${\sf S}_r$ initialized to $\ket{0}$.

  \noindent\hspace{6mm}
  2.2. Apply the Hadamard gate $H$ to ${\sf S}_r$.
  
  \noindent\hspace{6mm}
  2.3. (Controlled SWAP) If the content of ${\sf S}_r$ is $1$, swap ${\sf R}_{i,r}$ and ${\sf R}_{j,r}$.

  \noindent\hspace{6mm} 
  2.4. Apply the Hadamard gate $H$ and then the X gate on ${\sf S}_r$.

  \vspace{2mm}
  \noindent\hspace{3mm} 
  3. Compute the AND of all the registers ${\sf S}_1$,$\ldots$,${\sf S}_t$ in a new 1-qubit register ${\sf S}$.

  \vspace{2mm}
  \noindent\hspace{3mm} 
  4. Measure Register ${\sf S}$ in the computational basis.
  \vspace{2mm}
   \end{minipage}
  }
\end{center}\vspace{-2mm}
\caption{Description of the referee's procedure for Lemma \ref{lemma:MEQ}.}\label{fig:verif}
\end{figure}

The referee implements the procedure of Figure \ref{fig:verif}. Note that the conversion at Step 1 can be done locally by the referee since the referee knows $y$ and $z$ (remember that we are using linear quantum fingerprints, for which such a conversion is possible). Also note that Step 2 essentially implements, for each $r$, the SWAP test on registers $({\sf R}_{i,r},{\sf R}_{j,r})$. The only difference with the SWAP test described in Section~\ref{sec:prelim} (and Appendix \ref{Appendix:SWAP}) is that Register ${\sf S}_r$ is not measured. Instead, the AND of all the Registers ${\sf S}_1$,$\ldots$,${\sf S}_t$ is computed in a new register, which is then measured.
 
\subparagraph*{Analysis of the referee's procedure.} 
We now analyze the success probability of the procedure. First assume that $\MEQ_{k,n}(i, j, y, z)=1$. Then from Lemma \ref{lem:swap-test} we know that each SWAP test would output 1 with probability 1. This means that at Step 4, the measurement outcome is 1 with probability 1. 

Now assume that $\MEQ_{k,n}(i, j, y, z)=0$. 
By Lemma~\ref{lem:swap-test}, for each $r\in\{1,\ldots,t\}$,
the SWAP test on Registers $({\sf R}_{i,r},{\sf R}_{j,r})$, which has input $(|\Psi_{x_i}\rangle,|\Psi_{x_j}\rangle)$, 
would then output 1 with probability at most $\frac{5}{8}$.
Thus at Step 4, the measurement outcome is 0 with probability at least
$
1-\left(\frac{5}{8}\right)^t\ge 1-\varepsilon,
$
where the inequality follows from our choice of $t$.

In both cases we thus have success probability at least $1-\varepsilon$, as desired.
\end{proof}

\subsection{Implementing an MEQ Decision Tree: Proof of Theorem \ref{th:main}}\label{sub:tree}
\subparagraph*{MEQ decision trees.}
We define the model of Modified Equality Query decision trees ($\MEQ_{k,n}$ decision trees, or $\MEQ$ decision trees when the parameters $k, n$ are clear from the context) and the computation associated with them as follows.

\begin{itemize}
\item The input consists of $k$ $n$-bit strings $X_1,\cdots,X_k$; the output is an element of a set $S$.
\item The computational process is described by a binary tree $\mathcal{T}$ in which each node has either 0 or 2 children. 
Each internal node of the tree (i.e., each node with 2 children) is labeled by a 4-tuple $(i,j,y,z)$ for some indices $i,j\in\{1,\ldots,k\}$ and some (known) strings $y,z\in\{0,1\}^n$.
Each leaf (i.e., each node with 0 child) is labeled by an element in $S$.
\item The computation proceeds as follows. We start at the root. At each internal node we proceed to the right child if $X_i\oplus y =X_j \oplus z$ and to the left child if $X_i\oplus y \neq X_j \oplus z$. When reaching a leaf, we stop and output the label of the leaf. 
\end{itemize}

Observe that for any input $X_1,\ldots,X_k$, the above computational process can be implemented using at most $D$ modified equality queries, where $D$ denotes the depth of $\mathcal{T}$. For a function $f\colon(\{0,1\}^n)^k\to S$, we say that $\mathcal{T}$ computes $f$ if the output of the computational process induced by $\mathcal{T}$ is equal to $f(X_1,\ldots,X_k)$ for any $X_1,\ldots,X_k\in\{0,1\}^n$. 

\subparagraph*{ \boldmath{ Converting $\MEQ$ decision trees into $\SMP$ protocols.}}
Here is our main theorem (repeated from the introduction).
\addtocounter{theorem}{-4}
\begin{theorem}[repeated]
	For any $n, k, D \geq 0$ and $\delta \in (0,1)$,
	any $\MEQ_{k,n}$ decision tree of depth $D$ can be implemented
	by a quantum $\SMP$ protocol that uses
	$O(k  (\log D + \log(1/\delta))\log n)$ qubits
	and has error probability at most $\delta$.
\end{theorem}
\addtocounter{theorem}{4}

\begin{proof}
For each $r\in\{1,\dots,k\}$, $\ell$ sends to the referee the quantum state $F(x_\ell)$ specified by Lemma~\ref{lemma:MEQ} with 
$
\varepsilon=
  \frac{\delta}{4D}
.
$ 
The total communication cost is thus
$
O(k\log(1/\varepsilon)\log n)=O(k(\log D+\log(1/\delta)) \log n),
$
as claimed.

The referee then implements the computation induced by the $\MEQ_{k,n}$ decision tree, by starting from the root and then following the computational path. This requires making at most $D$ modified equality queries sequentially. 

Note that one individual modified equality query $\MEQ_{k,n}(i, j, y, z)$ can be implemented using the quantum states $F(x_i)$ and $F(x_j)$ received from player $i$ and player $j$. By using the procedure of Figure \ref{fig:verif} on these two quantum states, the success probability is at least $1- \frac{\delta}{4D}$. 

The main issue is that the procedure of Figure \ref{fig:verif} modifies the quantum states $F(x_i)$ and $F(x_j)$, which prevents reusing them for implementing the next modified equality query. To solve this issue, we convert the procedure of Figure~\ref{fig:verif} (which corresponds to a 2-outcome non-projective measurement) into a 2-outcome projective measurement using the conversion process mentioned in Section~\ref{sec:prelim} and described in Appendix~\ref{Appendix:pm}. Since this conversion preserves the success probability of the measurement, the success probability of the 2-outcome projective measurement that we obtain is at least $1-\frac{\delta}{4D}$. We implement the modified equality queries on the computation path by applying the corresponding 2-outcome projective measurements sequentially. Theorem \ref{thm:Gao15} shows that the overall success probability is at least
$
1-4D\left(1-\left(1-\frac{\delta}{4D}\right)\right)= 1-\delta,
$
as claimed.
\end{proof}


\section{Applications}
\label{sec:app}
In this section we present several applications of our compiler, ranging from statistical problems to graph problems.

\subsection{Warm-up: Grouping By Equality}

The simplest application of our compiler is to efficiently group the players by input, so that all players with the same input are in the same group:
formally, the $\prob{GroupByEQ}_{k,n}$ problem requires the referee to output a partition $P_1,\ldots,P_{s}$
of $[k]$,
such that for every $i, j \in [k]$,
there is an index $t$ such that $i,j \in P_t$ if and only if $x_i = x_j$, where $x_1,\ldots,x_k\in\{0,1\}^n$ are the players' inputs.

\begin{theorem}\label{thm:geq}
There exists a bounded-error $\SMP$
quantum protocol for $\prob{GroupByEQ}_{k,n}$ with communication cost 
$O(k\log k\log n)$.
\end{theorem}
\begin{proof}
	The $\prob{GroupByEQ}_{k,n}$ problem can be solved by an $\MEQ_{k,n}$ decision tree
	of depth $\binom{k}{2}$,
	where on each path we compare players' inputs against one another until we arrive at the correct output partition.
	See Figure~\ref{fig:geq} for an example with $k = 3$ players.
	Note that not all paths have the same length, as sometimes we can deduce the answer without comparing all inputs against one another; for example, if we learn that $x_1 = x_2$, then we no longer need to compare $x_2$ against the other inputs, as the query answers we obtain for $x_1$ imply the answers for $x_2$.
	The longest path is the leftmost path, where all queries return~0 (``not equal''), and the length of this path is exactly $\binom{k}{2}$.

  The conclusion then follows from Theorem~\ref{th:main}. 
\end{proof}

	\begin{figure}
		\begin{center}
			\begin{tikzpicture}[scale=0.95,every node/.style={align=center}]
			\tikzset{level distance=40pt,sibling distance=20pt}
				\Tree [.\node[draw]{$x_1 = x_2$?};
					      \edge node[auto=right] {$0$};  
					      [.\node[draw]{$x_1 = x_3$?};
					      		\edge node[auto=right] {$0$};  
						      [.\node[draw]{$x_2 =  x_3$?};
					      			\edge node[auto=right] {$0$};  
								[.\node[draw,font=\scriptsize]{$\set{1},\set{2},\set{3}$}; ]
					      			\edge node[auto=left] {$1$};  
								[.\node[draw,font=\scriptsize]{$\set{1},\set{2,3}$}; ]
							]
					      		\edge node[auto=left] {$1$};  
							[.\node[draw,font=\scriptsize]{$\set{1,3}, \set{2}$};]
						] 
					      \edge node[auto=left] {$1$};  
					      [.\node[draw]{$x_1 = x_3$?};
					      		\edge node[auto=right] {$0$};  
							[.\node[draw,font=\scriptsize]{$\set{1,2},\set{3}$};]
					      		\edge node[auto=left] {$1$};  
							[.\node[draw,font=\scriptsize]{$\set{1,2,3}$};]
						] 
					]
\end{tikzpicture}
\end{center}
\caption{An $\MEQ_{k,n}$ decision tree for $\prob{GroupByEQ}_{k,n}$ with $k = 3$ players. Each inner node is labeled with a query of the form ``$x_i = x_j$?'', which is short-hand notation for $\MEQ_{k,n}(i, j, 0^n, 0^n)$. The leaves are labeled with output partitions.}
\label{fig:geq}
\end{figure}
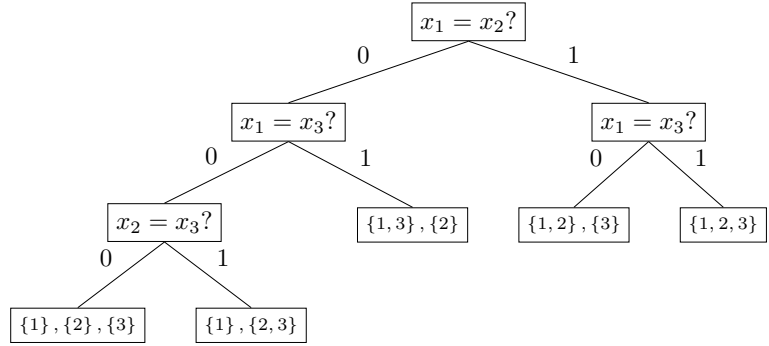

Using our protocol for $\prob{GroupByEQ}_{k,n}$ we can immediately solve several related problems.
First, we use it to solve the $\prob{AllEQ}_{k,n}$ and $\prob{ExistsEQ}_{k,n}$ problems, which ask us to determine whether all inputs are the same, or whether there exist two players that have the same input, respectively.
%
Ref.~\cite{FOZ20} showed that for any constant $\varepsilon>0$, 
the classical communication costs of $\prob{AllEQ}_{k,n}$ and $\prob{ExistsEQ}_{k,n}$ in the private-coin $\SMP$ model
are $\tilde{\Theta}(\sqrt{kn} + k)$ and $\tilde{\Theta}(k\sqrt{n})$, respectively.
Both problems reduce trivially to $\prob{GroupByEQ}_{k,n}$;
therefore Theorem~\ref{thm:geq} implies a quantum $\SMP$ protocol with communication cost $O(k \log k \log n)$
for both problems,\footnote{We remark that for $\mathrm{AllEQ}_{k,n}$, this can be further improved to $O(k\log n)$ by using the permutation test \cite{BBD+97,BCWW01,KNY08} instead of the SWAP test that we use here.}
 which is an exponential improvement in the dependence on $n$.
%
 More generally, for any $p \geq 0$,
	 we can compute the $p$-th frequency moment of the input,
		 $F_p = \sum_{w \in \set{0,1}^n} \left( f_w \right)^p$,
	 where $f_w$ is the frequency of the string $w$ in the input (i.e., the number of players whose input is $w$).
	 The case $p = 0$ corresponds to counting the number of distinct inputs.

	 Finally, we can use our protocol for grouping by equality
	 to solve
	 $P_3$-induced subgraph freeness:
		this problem
is set in the NIH network model (as explained in Section~\ref{sec:prelim}),
and requires us to determine whether the input graph contains an induced path consisting of two edges ($P_3$).
		As observed in~\cite{KMRS15}, a graph $G$ is $P_3$-induced subgraph free
		if and only if~$G$ is a collection of node-disjoint cliques.
		This can be tested by grouping the nodes of the graph
		using the input $\nu_v \oplus e_v$ (that is, the characteristic vector of $N(v) \cup \set{v}$)
		for each node $v$,
		and then checking if for each node $v \in V$,
		the number of nodes grouped together with $v$ (excluding $v$ itself)
		is exactly $\deg(v)$.
		To implement this test, we have each node send its degree to the referee,
		and then apply our protocol for $\prob{GroupByEQ}_{k,k}$
		to the vectors $\set{ \nu_v \oplus e_v }_{v \in V}$.
		The total communication cost is $O(k \log^2 k)$ qubits, 
		nearly matching the cost of the \emph{public coin} classical protocol from~\cite{KMRS15}.

\subsection{Neighborhood Diversity}

Our next application is to computing \emph{neighborhood diversity}~\cite{Lam12},
a graph parameter that is used in fixed-parameter tractability to measure the density of a graph
(in the same way that \emph{treewidth}, \emph{cliquewidth}, and other parameters are sometimes used).

The following definition is stated in the terminology of twins, for the sake of consistency
with the remainder of the paper, although this is not the terminology used in~\cite{Lam12}:
\begin{definition}[\cite{Lam12}]
	A graph $G = (V,E)$ has \emph{neighborhood diversity $d$}
	if its nodes can be partitioned into $d$ sets
	but no fewer,
	such that all nodes in each set are twins (false or true) of one another.
\end{definition}

We design an efficient quantum protocol for computing the neighborhood diversity of a graph,
based on Proposition~\ref{prop:twins} from Section~\ref{sec:prelim} (which is similar to what is used in, e.g.,~\cite{MPRT20}):

\begin{theorem}
  In the NIH network model, 
  there exists a bounded-error quantum protocol for computing the neighborhood diversity
  with communication cost $O(k\log^2 k)$.
\end{theorem}
\begin{proof}
	Proposition~\ref{prop:twins} 
	shows that nodes $v, w \in V$ are twins (false or true) if and only if
	$\nu_v = \nu_w$ or $\nu_v \oplus e_v = \nu_w \oplus e_w$.
	To compute the neighborhood diversity of $G$,
	we group nodes into sets of twins in much the same way that we used to solve $\prob{GroupByEQ}_{k,n}$ above,
	except that to determine whether two nodes $v,w$ are twins
	we need two queries:
	$\MEQ_{k,k}(\nu_v, \nu_w, 0^k, 0^k)$
	and $\MEQ_{k,k}(\nu_v, \nu_w, e_v, e_w)$.
	The resulting decision tree has depth $2\binom{k}{2}$,
	and the protocol then follows by Theorem~\ref{th:main}.
\end{proof}

We remark that computing neighborhood diversity can actually be done using ``plain'' equality queries alone (we do not need modified equality queries):
we can implement the tests above by having each node $v \in V$ send the referee a quantum fingerprint of its neighborhood $\nu_v$,
and also of $\nu_v \oplus e_v$.
However, this would fall outside the framework for our compiler, so it is simpler here to use modified equality queries and apply Theorem~\ref{th:main}.

\subsection{Reconstruction of Distance-Hereditary Graphs}

The \emph{reconstruction} task in the NIH network model requires the referee to output the entire input graph $G$.
For information-theoretic reasons, general graphs require a total of $\Theta(k^2)$ communication to reconstruct,
as this is the number of bits needed to represent an arbitrary graph on $k$ nodes.
However, for special classes of graphs, we can sometimes do much better:
for example,~\cite{KMRS15} showed that there is an efficient classical public-coin $\SMP$ protocol 
which reconstructs the input graph if it is $P_4$-induced subgraph free,
and rejects if it is not $P_4$-induced subgraph free.
This was generalized in~\cite{MPRT20} to distance-hereditary and bounded modular-width graphs.
In this subsection, we show that there is an efficient quantum private-coin $\SMP$ protocol for reconstructing distance-hereditary graphs.
The protocols of~\cite{KMRS15,MPRT20} for reconstructing $P_4$-induced subgraph free graphs
and bounded modular-width graphs can be adapted in a similar manner.\vspace{-3mm}

\subparagraph*{Distance-hereditary graphs and their properties.}
A graph $G=(V,E)$ is called {\em distance-hereditary} if the distance between two nodes
$v,w$ belonging to the same connected component
in $G$ is preserved 
in any induced subgraph of $G$ that contains $v$ and $w$.
Distance-hereditary graphs are characterized by the existence of a decomposition called a \emph{twin-pendant node decomposition}---a sequence $(v_1,\cdots,v_k)$ of nodes of $G$,
such that for each $j\in [k-1]$, one of the following conditions is true:
\begin{enumerate}[(C1)]
\item $v_j$ is a pendant node in $G[\{v_{j},\cdots,v_k\}]$ 
\item $v_j$ has a true twin in $G[\{v_{j},\cdots,v_k\}]$.
\item $v_j$ has a false twin in $G[\{v_{j},\cdots,v_k\}]$.
\end{enumerate}

It is known that a graph is distance-hereditary if and only if it has a twin-pendant decomposition (see, e.g., \cite{Bandelt1986}).
Moreover, if the graph is distance-hereditary, then the twin-pendant decomposition can be computed by repeatedly
choosing an arbitrary node satisfying one of the three conditions and removing it from the graph~\cite{KMRS15}.
This forms the basis for the reconstruction protocol given in~\cite{KMRS15,MPRT20}.
The protocol of~\cite{KMRS15,MPRT20} is stated in terms of polynomials, but we observe that it actually relies on 
simple algebraic properties, and can be translated to work with binary strings (interpreted as vectors over the binary field $\F_2$), as we do next.
(This abstracts and simplifies the algorithm of~\cite{KMRS15,MPRT20}.)\vspace{-3mm}

\subparagraph*{An algebraic characterization of pendant nodes and twins.}
The key to reconstructing distance-hereditary graphs is to find a representation of the graph that allows
us to repeatedly:
\begin{enumerate}[(1)]
	\item 
		\emph{Find} a node satisfying one of the three conditions (C1)--(C3),
		and
	\item \emph{Remove} this node from the graph and update our representation accordingly.
\end{enumerate}

The requirement of Definition \ref{def:valid} below is an adaptation and simplification of a corresponding requirement from~\cite{KMRS15,MPRT20}.
It requires that the graph be represented by a collection of linearly-independent vectors,
one for each node,
such that the neighborhood of each node is the sum of the representations of its neighbors,
allowing us to later ``peel off'' nodes from the graph by ``subtracting'' their representations.

\begin{definition}\label{def:valid}
Let $G=(V,E)$ be a $k$-node graph and $\ell$ be a positive integer.
A family of vectors $m=\set{(a_v,b_v)}_{v\in V}$, where $a_v,b_v\in\F_2^{\ell}$ for each $v\in V$, is
\emph{a valid representation of $G$} (or \emph{valid for $G$}, for short)
if:
\begin{enumerate}
	\item $\set{ a_v }_{v \in V}$ are linearly independent over $\F_2^{\ell}$,
		that is, there is no non-empty subset $U \subseteq V$ such that $\bigoplus_{u \in U} a_u = 0^\ell$;
		and
	\item For each $v \in V$, we have
			$b_v = \bigoplus_{u \in N(v)} a_u.$
\end{enumerate}
\end{definition}

The linear independence requirement of Definition~\ref{def:valid} leads to an algebraic characterization of the concepts of pendant nodes and twins which will be crucial to our algorithm:
\begin{proposition}
	If $\set{ (a_v, b_v) }_{v \in V}$ is a valid representation of $G = (V,E)$,
	then for every two nodes $w \neq u$ in $G$,
\begin{enumerate}
	\item Node $w$ is pendant and has node $u$ as its only neighbor if and only if $b_w = a_u$;
	\item Nodes $w,u$ are false twins if and only if $b_w = b_u$;
	\item Nodes $w,u$ are true twins if and only if $b_w \oplus a_w = b_u \oplus a_u$.
\end{enumerate}
\label{prop:test}
\end{proposition}
	This follows from the following property of linearly-independent sets:
	if $\set{ a_v }_{v \in V}$ is a linearly-independent set of vectors over $\F_2^{\ell}$,
	then
	for any two sets $S, T \subseteq V$
	we have 
	$\bigoplus_{u \in S} a_u = \bigoplus_{u \in T} a_u$
	if and only if
	$S = T$.

Our algorithm will work with a valid representation of the graph,
and modify it as it goes along.
The initial representation we will use is the following:
\begin{proposition}
	The representation $\set{ (e_v, \nu_v) }_{v \in V}$ is valid for $G = (V,E)$.
	\label{prop:initial}
\end{proposition}
\begin{proof}
	Indeed, $\set{ e_v}_{v \in V}$ are linearly independent,
	and $\nu_v = \bigoplus_{u \in N(v)} e_u$ for each $v \in V$.
\end{proof}
 
Next we show how to modify a valid representation after removing a node $w$ from the graph,
so that we obtain a valid representation for the remainder of the graph (the proof is very similar to \cite{KMRS15,MPRT20}, and is omitted here):
\begin{lemma}
	Let $m= \set{(a_v,b_v)}_{v\in V}$ be valid for $G=(V,E)$
	and let $w \neq u$ be nodes in $G$.
	We can obtain a valid representation $m'= \set{(a'_v,b'_v)}_{ v\in V\setminus\{w\} }$ for $G - w$
	as follows:
	\begin{enumerate}[I.]
		\item If $w$ is pendant and $u$ is its only neighbor: then for all $v \in V \setminus \set{w}$,
\[
a'_v=a_v
\hspace{5mm}
\textrm{and}
\hspace{5mm}
b'_v=\left\{
\begin{array}{ll}
b_v & (v\neq u),\\
b_u\oplus a_w & (v=u).
\end{array}
\right.
\]  
		\item If $w, u$ are false twins: then for all $v \in V \setminus \set{w}$,
\[
a'_v=\left\{
\begin{array}{ll}
a_v & (v\neq u),\\
a_u\oplus a_w & (v=u),
\end{array}
\right.
\hspace{5mm}
\textrm{and}
\hspace{5mm}
b'_v=b_v.
\]
		\item If $w, u$ are true twins: then for all $v \in V \setminus \set{w}$,
\[
a'_v=\left\{
\begin{array}{ll}
a_v & (v\neq u),\\
a_u\oplus a_w & (v=u),
\end{array}
\right.
\hspace{5mm}
\textrm{and}
\hspace{5mm}
b'_v=\left\{
\begin{array}{ll}
b_v & (v\neq u),\\
b_u\oplus a_w & (v=u).
\end{array}
\right.
\] 
	\end{enumerate}

	\label{lemma:update}
\end{lemma}

Together, Propositions~\ref{prop:test}, \ref{prop:initial} and Lemma~\ref{lemma:update}
give rise to the following abstract protocol for computing a pendant-twin decomposition (or identifying that
the graph is not distance-hereditary).
We cannot efficiently implement this protocol in the $\SMP$ model, as it requires players to send very long messages,
but we will show that we can \emph{simulate} it using modified equality queries.
From the twin-pendant decomposition output by Algorithm~\ref{alg:abstract}
it is easy to reconstruct the entire graph (as in~\cite{KMRS15,MPRT20}).

\begin{algorithm}
\caption{Abstract Protocol for Computing a Pendant-Twin Decomposition}\label{alg:abstract}

\KwIn{The representation $\set{ (e_v, \nu_v) }_{v \in V}$ of the graph $G$}

Set $a_v \leftarrow e_v, b_v \leftarrow \nu_v$ for each $v \in V$ \;
Set $\mathit{decomp} \leftarrow \lambda$ (an empty sequence) \;
\While{$|V| \geq 2$}
{
	\uIf(\tcp*[f]{Pendant $w$ with neighbor $u$}){$\exists w, u$ such that $b_w = a_u$}
	{
		\label{line:pendant}
		Append $(\text{``pendant''}, w, u)$ to $\mathit{decomp}$ \;
		Apply update (I) from Lemma~\ref{lemma:update} and remove $w$ from $V$
	}
	\uElseIf(\tcp*[f]{False twins $w,u$}){$\exists w,v$ such that $b_w = b_u$}
	{
		\label{line:false-twins}
		Append $(\text{``false twin''}, w, u)$ to $\mathit{decomp}$ \;
		Apply update (II) from Lemma~\ref{lemma:update} and remove $w$ from $V$
	}
	\uElseIf(\tcp*[f]{True twins $w,u$}){$\exists w,u$ such that $b_w \oplus a_w = b_u \oplus a_u$}
	{
		\label{line:true-twins}
		Append $(\text{``true twin''}, w, u)$ to $\mathit{decomp}$ \;
		Apply update (III) from Lemma~\ref{lemma:update} and remove $w$ from $V$
	}
	\lElse
	{
		Output ``Graph is not distance-hereditary''
	}
}
	Output $\mathit{decomp}$
\end{algorithm}

\subparagraph*{Simulating the abstract protocol using modified equality queries.}
As we said above, we cannot actually afford to implement Algorithm~\ref{alg:abstract}:
sending full neighborhood vectors $\nu_v$ requires $n$ qubits,
so we cannot even send the initial representation $\set{ (e_v, \nu_v) }_{v \in V}$ to the referee.
Instead, the referee works with \emph{fingerprints} of the neighborhood vectors, 
and we use modified equality queries to implement the tests in 
Lines~\ref{line:pendant},~\ref{line:false-twins} and~\ref{line:true-twins}.

To simulate the updates performed in Lemma~\ref{lemma:update},
we rely on the following crucial property:
upon removing node $w$,
we modify the representation $\set{ (a_v, b_v) }_{v \in V}$
by adding $a_w$ to some vectors and leaving the others unchanged.
By induction on the number of updates, we therefore have:
\begin{proposition}
	After performing $t \geq 0$
	steps resulting in a partial decomposition $\mathit{decomp}$,
	the resulting representation 
$\set{ (a_v^t, b_v^t) }_{v \in V}$ of the remaining graph
can be written in the form
	\begin{align*}
		& a_v^t = e_v \oplus \bigoplus_{u \in A_v^t} e_u
		\qquad
		\text{and}
		\qquad
		b_v^t = \nu_v \oplus \bigoplus_{u \in B_v^t} e_u,
	\end{align*}
	where $A_v^t, B_v^t \subseteq V$ depend only on $\mathit{decomp}$.
	\label{prop:mod}
\end{proposition}
This is important because the referee can explicitly construct 
	$\bigoplus_{u \in A_v^t} e_u$ and $\bigoplus_{u \in B_v^t} e_u$
	and use them as modifying vectors inside modified equality queries,
	as we show next:

\begin{theorem}
  In the NIH network model, 
  there is a bounded-error quantum protocol with communication cost $O(k \log^2 k)$ that enables the referee to reconstruct a distance-hereditary graph, or reject if the input graph is not distance-hereditary.
\end{theorem}
\begin{proof}
	The protocol is described in pseudocode in Algorithm~\ref{alg:concrete}.
	It is convenient to slightly abuse the notation by writing $\MEQ_{k,k}(i, y, z)$
	to denote the query ``$x_i \oplus y = z$?'', where $i \in [k]$ and $y, z \in \set{0,1}^k$.
	The referee can perform this query
	by computing a fingerprint for the vector $0^k$
	and then proceeding as shown in Section~\ref{sec:meq},
	using the fingerprint for player $i$'s input, the fingerprint for $0^k$, and the vectors $y,z$.

	In the protocol, the referee explicitly maintains the vectors $\set{ a_v }_{v \in V}$
	of the representation,
	and implicitly maintains the vectors $\set{ b_v }_{v \in V}$
	by storing \emph{modifier vectors} $\set{ c_v }_{v \in V}$,
	such that $b_v = \nu_v \oplus c_v$ for each $v \in V$
	(this is possible due to Proposition~\ref{prop:mod}).
	To simulate each test in Algorithm~\ref{alg:abstract}
	we use appropriate modified equality queries:
	for example, to simulate the test ``$b_w = a_u$?'' in line~\ref{line:pendant}
	of Algorithm~\ref{alg:abstract},
	we use the query $\MEQ_{k,k}(w, c_w, a_u)$ in line~\ref{line:meq-pendant} of Algorithm~\ref{alg:concrete},
	which checks whether $\nu_w \oplus c_w = a_u$.
	Since $b_w = \nu_w \oplus c_w$, this corresponds to exactly the same test.
  
\begin{algorithm}
\caption{Quantum $\SMP$ Protocol for Computing a Pendant-Twin Decomposition}\label{alg:concrete}

Set $a_v \leftarrow e_v, c_v \leftarrow 0^k$ for each $v \in V$ \;
Set $\mathit{decomp} \leftarrow \lambda$ (an empty sequence) \;

\While{$|V| \geq 2$}
{
	\uIf(\tcp*[f]{\!\!Pendant $w$ with neighbor $u$}){$\exists w, u$ such that $\MEQ_{k,k}(w, c_w, a_u) = 1$}
	{
		\label{line:meq-pendant}
		Append $(\text{``pendant''}, w, u)$ to $\mathit{decomp}$ and set $V \leftarrow V \setminus \set{w}$ \;
		Set $c_u \leftarrow c_u \oplus a_w$ \;
	}
	\uElseIf(\tcp*[f]{False twins $w,u$}){$\exists w,v$ such that $\MEQ_{k,k}(w, u, c_w, c_u) = 1$}
	{
		\label{line:meq-false-twins}
		Append $(\text{``false twin''}, w, u)$ to $\mathit{decomp}$ and set $V \leftarrow V \setminus \set{w}$ \;
		Set $a_u \leftarrow a_u \oplus a_w$ \;
	}
	\uElseIf(\tcp*[f]{True twins $w,u$}){$\exists w,u$ s.t.\ $\MEQ_{k,k}(w, u,  c_w \oplus a_w,c_u \oplus a_u) = 1$}
	{
		\label{line:meq-true-twins}
		Append $(\text{``true twin''}, w, u)$ to $\mathit{decomp}$ and set $V \leftarrow V \setminus \set{w}$ \;
		Set $a_u \leftarrow a_u \oplus a_w$ and $c_u \leftarrow c_u \oplus a_w$ \;
	}
	\lElse
	{
		Output ``Graph is not distance-hereditary''
	}
}
	Output $\mathit{decomp}$
\end{algorithm}
Each of the tests in lines~\ref{line:meq-pendant},~\ref{line:meq-false-twins} and~\ref{line:meq-true-twins}
of Algorithm~\ref{alg:concrete}
can be implemented using $\binom{k}{2}$ modified equality queries.
The other steps do not require any queries.
The whole procedure can thus be implemented by an $\MEQ_{k,k}$ decision tree of depth $(k-1)3\binom{k}{2}$;
  the quantum protocol can be constructed from Theorem~\ref{th:main}.
\end{proof}

\subsection{Enumeration of Isolated Cliques}
Finally, we turn our attention to the problem of enumerating all \emph{isolated cliques}
in a graph.
Isolated cliques and pseudocliques are important concepts in complex network analysis (see, e.g., \cite{IIO05,II09,KHMN09}).
Concretely, we show how to enumerate max-$d$-isolated cliques:

\begin{definition}[\cite{KHMN09}, Definition 3] 
	A subgraph $S$ of $G=(V,E)$ is a {\em max-$d$-isolated clique} 
if the subgraph induced by $S$ is a clique, and each node in $S$ has at most $d$ edges to $V \setminus S$.
\end{definition}

It is convenient to describe our protocol for enumerating max-$d$-isolated cliques using a new primitive
that we call \emph{modified bounded Hamming distance queries},
which compute the Hamming distance between two modified inputs,
but only if the distance does not exceed some fixed threshold $d$:
\begin{equation*}
  \prob{MHAM}_{n}^d(i,j,y,z)=
  \left\{
    \begin{array}{ll}
	    \Delta_n(x_i \oplus y, x_j \oplus z) & \textrm{if } \Delta_n(x_i\oplus y,x_j\oplus z)\le d,\\
    \perp & \textrm{otherwise}.
    \end{array}
    \right.
   \end{equation*}
   An $\prob{MHAM}_n^d$ query can be computed by an $\MEQ_{k,n}$ decision tree of depth $\sum_{c = 0}^d \binom{n}{c} = O(n^d)$,
   which checks, for all strings $e \in \set{0,1}^n$ of Hamming weight $c \leq d$,
   whether $\left( x_i \oplus y \right) \oplus e = x_j \oplus z$.%
   \footnote{
	   Yao \cite{Yao03} showed that for any constant $d$,
	one can test whether $\Delta_n(x,y) \leq d$
	   using $O(\log n)$ qubits in the two-party quantum $\SMP$
	   model.
	   However, unlike the equality function, it is not clear how to convert this protocol
	   into a protocol for \emph{modified} Hamming distance queries, as it does not have the linearity property
   that was crucial to prove Lemma~\ref{lemma:MEQ}.}
   %
   We note that the restriction to a fixed upper bound $d$ is important, as in general,
   computing the exact Hamming distance between two strings requires linear communication,
   even for interactive quantum protocols where the parties have shared entanglement~\cite{HSZZ06}.%
   \footnote{Specifically,~\cite{HSZZ06} shows that testing whether two $n$-bit strings have Hamming 
	   distance at most $d$ requires $\Omega(d)$ bits of communication, for any $d \leq n/2$;
   this implies that computing the exact Hamming distance between general $n$-bit strings
   requires $\Omega(n)$ qubits.}


   Our protocol for enumerating max-$d$-isolated cliques in the NIH network model will use modified bounded Hamming distance queries with $n=k$. The protocol is motivated by the following observations.
   The first observation allows us to use \emph{bounded} Hamming distance queries,
   as it shows that we do not need to worry about nodes with $\Delta_k(\nu_u \oplus e_u, \nu_v\oplus e_v) > 2d$:
   \begin{proposition}
	   If $S \subseteq V$ contains two nodes $u \neq v$ such that $\Delta_k(\nu_u \oplus e_u, \nu_v\oplus e_v) > 2d$,
				then $S$ cannot be a max-$d$-isolated clique.
				\label{prop:2d}
   \end{proposition}
   A symmetric difference of cardinality $>2d$ between $N(u) \cup \set{u}$ and $N(v) \cup \set{v}$
   implies that either 
   $N(u) \cup \set{u}$ contains more than $d$ nodes that are not in $N(v) \cup \set{v}$,
   or vice-versa.
   If $S$ is a clique, this means that either $u$ or $v$ has more than $d$ neighbors outside $S$,
   so $S$ is not a max-$d$-isolated clique;
   if $S$ is not a clique, then in particular it is not a max-$d$-isolated clique.

   The next observation is useful for checking whether a given set is a clique or not:
   \begin{proposition}
	   For every $u,v \in V$,
	   if $\set{u,v} \notin E$ then
	   $\Delta_k(\nu_u, \nu_v) = \Delta_k(\nu_u \oplus e_u, \nu_v \oplus e_v)-2$,
	   and
	   if $\set{u,v} \in E$ then
	   $\Delta_k(\nu_u, \nu_v) = \Delta_k(\nu_u \oplus e_u, \nu_v \oplus e_v) + 2$.
	\label{prop:edge_test}
   \end{proposition}
   This is because if $\set{u,v} \notin E$,
   then $N(u) \ominus N(v) = \left(\left( N(u) \cup \set{u} \right) \ominus \left( N(v) \cup \set{v} \right)\right)\setminus\{u,v\}$
   (where $\ominus$ denotes the symmetric difference),
   whereas if $\set{u,v} \in E$,
   then
   $N(u) \ominus N(v) = \set{u,v} \cup  \left( N(u) \cup \set{u} \right) \ominus \left( N(v) \cup \set{v} \right)$.

   One implication of Proposition~\ref{prop:edge_test}
   is that we always have 
	   $\Delta_k(\nu_u, \nu_v) \leq \Delta_k(\nu_u \oplus e_u, \nu_v \oplus e_v) + 2$.
	   Together with Proposition~\ref{prop:2d},
	   this gives us an upper bound of $2d+2$ on the Hamming distance $\Delta_k(\nu_u, \nu_v)$
	   for nodes that might belong to the same max-$d$-isolated clique.

Our main result for computing max-$d$-isolated cliques is as follows:
\begin{theorem}
  In the NIH network model, 
  for any $d\ge 1$,
  there is a bounded-error quantum $\SMP$ protocol with communication cost $O(kd \log^2 k)$
  for enumerating all the max-$d$-isolated cliques of the input graph.
  \label{thm:cliques}
\end{theorem}
 
\begin{proof}
	We compute $\prob{MHAM}_k^{2d+2}(u,v,0^k,0^k)$ and
	$\prob{MHAM}_k^{2d}(u,v,e_u,e_v)$ for all pairs $u,v\in G$,
	using an $\MEQ_k$ decision tree of depth $\binom{k}{2} \cdot O(n^d) + \binom{k}{2} \cdot O(n^{d+2}) = O(k^2 n^{d+2})$
	(this consists of $2\binom{k}{2}$ ``small'' $\MEQ$ decision trees, one for each $\prob{MHAM}$
	query, composed with one another).
	In addition, we have each node send its degree to the referee.

	We label each leaf of the decision tree based on the results of the $\prob{MHAM}$ queries
	leading to the leaf.
	Each subset $S \subseteq V$ is listed as a max-$d$-isolated clique in a given leaf if it satisfies
	the following three conditions:
\begin{enumerate}
	\item $\prob{MHAM}_k^{2d}(u, v, e_u,e_v) \neq \bot$ for all pairs $u,v\in S$.
		By Proposition~\ref{prop:edge_test},
		this also implies that $\prob{MHAM}_k^{2d+2}(u, v, 0^k, 0^k) \neq \bot$ for all $u,v\in S$.
	\item $\prob{MHAM}_k^{2d+2}(u,v, 0^k, 0^k)= \prob{MHAM}_k^{2d}(u, v, e_u,e_v)+2$ for all pairs $u,v\in S$.
		And finally,
  \item $\deg(u)\le |S| + d -1$ for all $u\in S$.
\end{enumerate}


The correctness of this enumeration algorithm follows from the two observations above,
together with the fact that if $S$ is a clique, then a node $u \in S$
has at most $d$ edges going outside $S$ if and only if $\deg(u) \leq |S| + d - 1$.
\end{proof}

\bibliographystyle{plainurl}
\bibliography{refs}

\providecommand{\noopsort}[1]{}
\begin{thebibliography}{10}

\bibitem{Aar2016}
Scott Aaronson.
\newblock Lecture notes for the 28th {McGill} invitational workshop on computational complexity.
\newblock Arxiv:1607.05256, 2016.

\bibitem{A+24}
Amirreza Akbari, Xavier Coiteux-Roy, Francesco d'Amore, Fran{\c c}ois {Le Gall}, Henrik Lievonen, Darya Melnyk, Augusto Modanese, Shreyas Pai, Marc-Olivier Renou, V{\' a}clav Rozho{\v n}, and Jukka Suomela.
\newblock Online locality meets distributed quantum computing.
\newblock ArXiv:2403.01903, 2024.

\bibitem{AVPODC22}
Joran~{\noopsort{Apeldoorn}}van Apeldoorn and Tijn de~Vos.
\newblock A framework for distributed quantum queries in the {CONGEST} model.
\newblock In {\em Proceedings of the 2022 {ACM} Symposium on Principles of Distributed Computing (PODC 2022)}, pages 109--119, 2022.
\newblock \href {https://doi.org/10.1145/3519270.3538413} {\path{doi:10.1145/3519270.3538413}}.

\bibitem{Arfaoui+14}
Heger Arfaoui and Pierre Fraigniaud.
\newblock What can be computed without communications?
\newblock {\em {SIGACT} News}, 45(3):82--104, 2014.
\newblock \href {https://doi.org/10.1145/2670418.2670440} {\path{doi:10.1145/2670418.2670440}}.

\bibitem{Babai+97}
L\'aszl\'o Babai and Peter~G. Kimmel.
\newblock Randomized simultaneous messages: Solution of a problem of {Yao} in communication complexity.
\newblock In {\em Proceedings of the 12th Annual IEEE Conference on Computational Complexity (CCC 1997)}, pages 239--246, 1997.
\newblock \href {https://doi.org/10.1109/CCC.1997.612319} {\path{doi:10.1109/CCC.1997.612319}}.

\bibitem{Bandelt1986}
Hans{-}J{\"{u}}rgen Bandelt and Henry~Martyn Mulder.
\newblock Distance-hereditary graphs.
\newblock {\em Journal of Combinatorial Theory, Series B}, 41(2):182--208, 1986.
\newblock \href {https://doi.org/10.1016/0095-8956(86)90043-2} {\path{doi:10.1016/0095-8956(86)90043-2}}.

\bibitem{BBD+97}
Adriano Barenco, Andr{\'{e}} Berthiaume, David Deutsch, Artur Ekert, Richard Jozsa, and Chiara Macchiavello.
\newblock Stabilization of quantum computations by symmetrization.
\newblock {\em {SIAM} Journal on Computing}, 26(5):1541--1557, 1997.
\newblock \href {https://doi.org/10.1137/S0097539796302452} {\path{doi:10.1137/S0097539796302452}}.

\bibitem{BMRT14}
Florent Becker, Pedro Montealegre, Ivan Rapaport, and Ioan Todinca.
\newblock The simultaneous number-in-hand communication model for networks: Private coins, public coins and determinism.
\newblock In {\em Proceedings of the 21st International Colloquium on Structural Information and Communication Complexity (SIROCCO 2024)}, volume 8576 of {\em Lecture Notes in Computer Science}, pages 83--95. Springer, 2014.
\newblock \href {https://doi.org/10.1007/978-3-319-09620-9\_8} {\path{doi:10.1007/978-3-319-09620-9\_8}}.

\bibitem{BCWW01}
Harry Buhrman, Richard Cleve, John Watrous, and Ronald {de Wolf}.
\newblock Quantum fingerprinting.
\newblock {\em Physical Review Letters}, 87:167902, 2001.
\newblock \href {https://doi.org/10.1103/PhysRevLett.87.167902} {\path{doi:10.1103/PhysRevLett.87.167902}}.

\bibitem{CR+STOC24}
Xavier Coiteux-Roy, Francesco {d'Amore}, Rishikesh Gajjala, Fabian Kuhn, Fran{\c c}ois {Le Gall}, Henrik Lievonen, Augusto Modanese, Marc-Olivier Renou, Gustav Schmid, and Jukka Suomela.
\newblock No distributed quantum advantage for approximate graph coloring.
\newblock In {\em Proceedings of the 56th ACM Symposium on Theory of Computing (STOC 2024)}, pages 1901--1910, 2024.
\newblock \href {https://doi.org/10.1145/3618260.3649679} {\path{doi:10.1145/3618260.3649679}}.

\bibitem{Denchev+08}
Vasil~S. Denchev and Gopal Pandurangan.
\newblock Distributed quantum computing: a new frontier in distributed systems or science fiction?
\newblock {\em SIGACT News}, 39(3):77--95, 2008.
\newblock \href {https://doi.org/10.1145/1412700.1412718} {\path{doi:10.1145/1412700.1412718}}.

\bibitem{ElkinKNP14}
Michael Elkin, Hartmut Klauck, Danupon Nanongkai, and Gopal Pandurangan.
\newblock Can quantum communication speed up distributed computation?
\newblock In {\em Proceedings of the 33rd {ACM} Symposium on Principles of Distributed Computing (PODC 2014)}, pages 166--175, 2014.
\newblock \href {https://doi.org/10.1145/2611462.2611488} {\path{doi:10.1145/2611462.2611488}}.

\bibitem{FOZ20}
Orr Fischer, Rotem Oshman, and Uri Zwick.
\newblock Public vs. private randomness in simultaneous multi-party communication complexity.
\newblock {\em Theoretical Computer Science}, 810:72--81, 2020.
\newblock \href {https://doi.org/10.1016/j.tcs.2018.04.032} {\path{doi:10.1016/j.tcs.2018.04.032}}.

\bibitem{FLNP21}
Pierre Fraigniaud, Fran{\c c}ois~Le Gall, Harumichi Nishimura, and Ami Paz.
\newblock Distributed quantum proofs for replicated data.
\newblock In {\em Proceedings of the 12th Innovations in Theoretical Computer Science Conference (ITCS 2021)}, pages 28:1--28:20, 2021.
\newblock \href {https://doi.org/10.4230/LIPICS.ITCS.2021.28} {\path{doi:10.4230/LIPICS.ITCS.2021.28}}.

\bibitem{Fraigniaud2024}
Pierre Fraigniaud, Mael Luce, Fr{\'{e}}d{\'{e}}ric Magniez, and Ioan Todinca.
\newblock Even-cycle detection in the randomized and quantum {CONGEST} model.
\newblock In {\em Proceedings of the 43rd {ACM} Symposium on Principles of Distributed Computing (PODC 2024)}, page 209–219, 2024.
\newblock \href {https://doi.org/10.1145/3662158.3662767} {\path{doi:10.1145/3662158.3662767}}.

\bibitem{Gao15}
Jingliang Gao.
\newblock Quantum union bounds for sequential projective measurements.
\newblock {\em Physical Review A}, 92:052331, 2015.
\newblock \href {https://doi.org/10.1103/PhysRevA.92.052331} {\path{doi:10.1103/PhysRevA.92.052331}}.

\bibitem{Gavinsky+CCC13}
Dmitry Gavinsky, Tsuyoshi Ito, and Guoming Wang.
\newblock Shared randomness and quantum communication in the multi-party model.
\newblock In {\em Proceedings of the 28th Conference on Computational Complexity ({CCC} 2013)}, pages 34--43, 2013.
\newblock \href {https://doi.org/10.1109/CCC.2013.13} {\path{doi:10.1109/CCC.2013.13}}.

\bibitem{GKW04}
Dmitry Gavinsky, Julia Kempe, and Ronald {de Wolf}.
\newblock Quantum communication cannot simulate a public coin, 2004.
\newblock ArXiv:quant-ph/0411051.

\bibitem{Gavinsky+2008}
Dmitry Gavinsky and Pavel Pudl{\'{a}}k.
\newblock Exponential separation of quantum and classical non-interactive multi-party communication complexity.
\newblock In {\em Proceedings of the 23rd Annual {IEEE} Conference on Computational Complexity ({CCC} 2008)}, pages 332--339, 2008.
\newblock \href {https://doi.org/10.1109/CCC.2008.27} {\path{doi:10.1109/CCC.2008.27}}.

\bibitem{GavoilleKM09}
Cyril Gavoille, Adrian Kosowski, and Marcin Markiewicz.
\newblock What can be observed locally?
\newblock In {\em Proceedings of the 23rd International Symposium on Distributed Computing (DISC 2009)}, volume 5805 of {\em LNCS}, pages 243--257. Springer, 2009.
\newblock \href {https://doi.org/10.1007/978-3-642-04355-0\_26} {\path{doi:10.1007/978-3-642-04355-0\_26}}.

\bibitem{GS20}
Hip\'olito G\'omez-Sousa.
\newblock Multi-party quantum fingerprinting with weak coherent pulses: circuit design and protocol analysis.
\newblock {\em New Journal of Physics}, 22:113004, 2020.
\newblock \href {https://doi.org/10.1088/1367-2630/abc2e5} {\path{doi:10.1088/1367-2630/abc2e5}}.

\bibitem{Hasegawa+PODC24}
Atsuya Hasegawa, Srijita Kundu, and Harumichi Nishimura.
\newblock On the power of quantum distributed proofs.
\newblock In {\em Proceedings of the 43rd {ACM} Symposium on Principles of Distributed Computing (PODC 2024)}, page 220–230, 2024.
\newblock \href {https://doi.org/10.1145/3662158.3662788} {\path{doi:10.1145/3662158.3662788}}.

\bibitem{HSZZ06}
Wei Huang, Yaoyun Shi, Shengyu Zhang, and Yufan Zhu.
\newblock The communication complexity of the {Hamming} distance problem.
\newblock {\em Information Processing Letters}, 99(4):149--153, 2006.
\newblock \href {https://doi.org/10.1016/j.ipl.2006.01.014} {\path{doi:10.1016/j.ipl.2006.01.014}}.

\bibitem{II09}
Hiro Ito and Kazuo Iwama.
\newblock Enumeration of isolated cliques and pseudo-cliques.
\newblock {\em {ACM} Transactions on Algorithms}, 5(4):40:1--40:21, 2009.
\newblock \href {https://doi.org/10.1145/1597036.1597044} {\path{doi:10.1145/1597036.1597044}}.

\bibitem{IIO05}
Hiro Ito, Kazuo Iwama, and Tsuyoshi Osumi.
\newblock Linear-time enumeration of isolated cliques.
\newblock In {\em Proceedings of the 13th Annual European Symposium (ESA 2005)}, volume 3669 of {\em Lecture Notes in Computer Science}, pages 119--130. Springer, 2005.
\newblock \href {https://doi.org/10.1007/11561071\_13} {\path{doi:10.1007/11561071\_13}}.

\bibitem{Izumi+PODC19}
Taisuke Izumi and Fran{\c c}ois~Le Gall.
\newblock Quantum distributed algorithm for the {All-Pairs Shortest Path} problem in the {CONGEST-CLIQUE} model.
\newblock In {\em Proceedings of the 38th {ACM} Symposium on Principles of Distributed Computing (PODC 2019)}, pages 84--93, 2019.
\newblock \href {https://doi.org/10.1145/3293611.3331628} {\path{doi:10.1145/3293611.3331628}}.

\bibitem{Izumi+STACS20}
Taisuke Izumi, Fran{\c{c}}ois {Le Gall}, and Fr{\'{e}}d{\'{e}}ric Magniez.
\newblock Quantum distributed algorithm for triangle finding in the {CONGEST} model.
\newblock In {\em Proceedings of the 37th International Symposium on Theoretical Aspects of Computer Science (STACS 2020)}, pages 23:1--23:13, 2020.
\newblock \href {https://doi.org/10.4230/LIPIcs.STACS.2020.23} {\path{doi:10.4230/LIPIcs.STACS.2020.23}}.

\bibitem{KNY08}
Masaru Kada, Harumichi Nishimura, and Tomoyuki Yamakami.
\newblock The efficiency of quantum identity testing of multiple states.
\newblock {\em Journal of Physics A: Mathematical and Theoretical}, 41:395309, 2008.
\newblock \href {https://doi.org/10.1088/1751-8113/41/39/395309} {\path{doi:10.1088/1751-8113/41/39/395309}}.

\bibitem{KMRS15}
Jarkko Kari, Mart{\'{\i}}n Matamala, Ivan Rapaport, and Ville Salo.
\newblock Solving the induced subgraph problem in the randomized multiparty simultaneous messages model.
\newblock In {\em Proceedings of the 22nd International Colloquium on Structural Information and Communication Complexity (SIROCCO 2015)}, volume 9439 of {\em Lecture Notes in Computer Science}, pages 370--384. Springer, 2015.
\newblock \href {https://doi.org/10.1007/978-3-319-25258-2\_26} {\path{doi:10.1007/978-3-319-25258-2\_26}}.

\bibitem{KHMN09}
Christian Komusiewicz, Falk H{\"{u}}ffner, Hannes Moser, and Rolf Niedermeier.
\newblock Isolation concepts for efficiently enumerating dense subgraphs.
\newblock {\em Theoretical Computer Science}, 410(38-40):3640--3654, 2009.
\newblock \href {https://doi.org/10.1016/j.tcs.2009.04.021} {\path{doi:10.1016/j.tcs.2009.04.021}}.

\bibitem{Lam12}
Michael Lampis.
\newblock Algorithmic meta-theorems for restrictions of treewidth.
\newblock {\em Algorithmica}, 64(1):19--37, 2012.
\newblock \href {https://doi.org/10.1007/s00453-011-9554-x} {\path{doi:10.1007/s00453-011-9554-x}}.

\bibitem{LeGall+2023}
Fran\c{c}ois Le~Gall, Masayuki Miyamoto, and Harumichi Nishimura.
\newblock {Distributed Merlin-Arthur Synthesis of Quantum States and Its Applications}.
\newblock In {\em Proceedings of the 48th International Symposium on Mathematical Foundations of Computer Science (MFCS 2023)}, pages 63:1--63:15, 2023.
\newblock \href {https://doi.org/10.4230/LIPIcs.MFCS.2023.63} {\path{doi:10.4230/LIPIcs.MFCS.2023.63}}.

\bibitem{LeGall+PODC18}
Fran{\c{c}}ois {Le Gall} and Fr{\'e}d{\'e}ric Magniez.
\newblock Sublinear-time quantum computation of the diameter in {CONGEST} networks.
\newblock In {\em In Proceedings of the 37th ACM Symposium on Principles of Distributed Computing (PODC 2018)}, pages 337--346, 2018.
\newblock \href {https://doi.org/10.1145/3212734.3212744} {\path{doi:10.1145/3212734.3212744}}.

\bibitem{LMN23}
Fran{\c c}ois Le~Gall, Masayuki Miyamoto, and Harumichi Nishimura.
\newblock Distributed quantum interactive proofs.
\newblock In {\em Proceedings of the 40th International Symposium on Theoretical Aspects of Computer Science (STACS 2023)}, pages 63:1--63:15, 2023.
\newblock \href {https://doi.org/10.4230/LIPICS.STACS.2023.42} {\path{doi:10.4230/LIPICS.STACS.2023.42}}.

\bibitem{LeGall+STACS19}
Fran{\c{c}}ois {Le Gall}, Harumichi Nishimura, and Ansis Rosmanis.
\newblock Quantum advantage for the {LOCAL} model in distributed computing.
\newblock In {\em Proceedings of the International Symposium on Theoretical Aspects of Computer Science (STACS)}, pages 49:1--49:14, 2019.
\newblock \href {https://doi.org/10.4230/LIPIcs.STACS.2019.49} {\path{doi:10.4230/LIPIcs.STACS.2019.49}}.

\bibitem{MPRT20}
Pedro Montealegre, Sebastian Perez{-}Salazar, Ivan Rapaport, and Ioan Todinca.
\newblock Graph reconstruction in the congested clique.
\newblock {\em Journal of Computer and System Sciences}, 113:1--17, 2020.
\newblock \href {https://doi.org/10.1016/j.jcss.2020.04.004} {\path{doi:10.1016/j.jcss.2020.04.004}}.

\bibitem{Newman91}
Ilan Newman.
\newblock Private vs. common random bits in communication complexity.
\newblock {\em Information Processing Letters}, 39(2):67--71, 1991.
\newblock \href {https://doi.org/10.1016/0020-0190(91)90157-D} {\path{doi:10.1016/0020-0190(91)90157-D}}.

\bibitem{Newman1996}
Ilan Newman and Mario Szegedy.
\newblock Public vs. private coin flips in one round communication games (extended abstract).
\newblock In {\em Proceedings of the Twenty-Eighth Annual {ACM} Symposium on the Theory of Computing (STOC 1996)}, pages 561--570, 1996.
\newblock \href {https://doi.org/10.1145/237814.238004} {\path{doi:10.1145/237814.238004}}.

\bibitem{Q+21}
Ji-Qian Qin, Jing-Tao Wang, Yun-Long Yu, and Xiang-Bin Wang.
\newblock General theory of quantum fingerprinting network.
\newblock {\em Physical Review Research}, 3:033039, 2021.
\newblock \href {https://doi.org/10.1103/PhysRevResearch.3.033039} {\path{doi:10.1103/PhysRevResearch.3.033039}}.

\bibitem{Tani+12}
Seiichiro Tani, Hirotada Kobayashi, and Keiji Matsumoto.
\newblock Exact quantum algorithms for the leader election problem.
\newblock {\em {ACM} Transactions on Computation Theory}, 4(1):1:1--1:24, 2012.
\newblock \href {https://doi.org/10.1145/2141938.2141939} {\path{doi:10.1145/2141938.2141939}}.

\bibitem{WYPODC22}
Xudong Wu and Penghui Yao.
\newblock Quantum complexity of weighted diameter and radius in {CONGEST} networks.
\newblock In {\em Proceedings of the 42nd {ACM} Symposium on Principles of Distributed Computing (PODC 2022)}, pages 120--130, 2022.
\newblock \href {https://doi.org/10.1145/3519270.3538441} {\path{doi:10.1145/3519270.3538441}}.

\bibitem{Yao03}
Andrew~Chi{-}Chih Yao.
\newblock On the power of quantum fingerprinting.
\newblock In {\em Proceedings of the 35th Annual {ACM} Symposium on Theory of Computing (STOC 2003)}, pages 77--81, 2003.
\newblock \href {https://doi.org/10.1145/780542.780554} {\path{doi:10.1145/780542.780554}}.

\end{thebibliography}

\clearpage
\appendix

\section{Illustration of a 2-Outcome Measurement}\label{Appendix:fig}

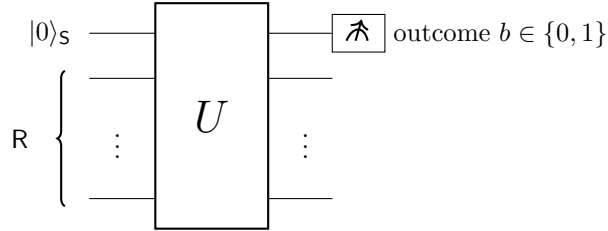
\begin{figure}[ht!]
    \centering
    \begin{tikzpicture}[scale=0.4,rectnode/.style={thick,shape=rectangle,draw=black,fill=white,minimum height=30mm, minimum width=15mm},rectnode2/.style={shape=rectangle,draw=black,minimum height=5mm, minimum width=7mm},roundnode/.style={circle, draw=black, fill=black, minimum height=1mm}]
        \newcommand\XA{9.6}
        \newcommand\YA{5}
            
        \node[] (ancilla) at (-1.5,1.7*\YA) {};
        \node[] (R1) at (-1.5,1.4*\YA) {};
        \node[] (R4) at (-1.5,0.6*\YA) {};
        \node[] (RR1) at (0.75*\XA,1.4*\YA) {};
        \node[] (RR4) at (0.75*\XA,0.6*\YA) {};
        \node[] (RR5) at (1.3*\XA,1.69*\YA) {outcome $b\in \{0,1\}$};

        \draw (R1) -- (RR1);
        \draw (R4) -- (RR4);

        \node[rectnode2] (mes) at (0.6*\XA+2,1.7*\YA) {};
        \draw (ancilla) -- (mes);
        \draw[->, thick] (0.6*\XA+1.9,1.65*\YA) -- (0.6*\XA+2.1,1.79*\YA);
        \draw[thick] (0.6*\XA+1.6,1.65*\YA) .. controls  (0.6*\XA+1.8,1.75*\YA) and  (0.6*\XA+2.2,1.75*\YA) .. (0.6*\XA+2.3,1.65*\YA);

         \node[rectnode] (U) at (0.3*\XA,1.15*\YA) {\LARGE $U$};
         
         \node[draw=none,fill=none] at (-2.5,1.7*\YA) {$\ket{0}_{\mathsf{S}}$};
         \draw[thick,decorate,decoration={brace,amplitude=1mm}] (-2,-1+0.75*\YA) -- (-2,3.5+0.75*\YA);
         \node[draw=none,fill=none] at (-3.5,1*\YA) {$\sf{R}$};
         \node[draw=none,fill=none] at (6.00,1*\YA) {$\vdots$};
         \node[draw=none,fill=none] at (-0.25,1*\YA) {$\vdots$};
    \end{tikzpicture}
    \caption{A 2-outcome measurement of a quantum register $\mathsf{R}$.}\label{fig:GM}
    \end{figure}

\section{Projective measurements}\label{Appendix:pm}
Projective measurements are a special kind of measurements allowed by the laws of quantum mechanics. We explain below this concept and how to convert a 2-outcome measurement (as introduced in Section \ref{sec:prelim}) into a 2-outcome projective measurement. 

The 2-outcome projective measurement corresponding to the 2-outcome measurement associated with the unitary $U$ is the measurement described in Figure \ref{fig:PM}. There are two key properties. First, for any $b\in\{0,1\}$, the probability of obtaining $b$ is exactly the same as the probability of obtaining $b$ in the standard measurement associated with $U$. The second, and crucial, property is that for any $b\in\{0,1\}$, if the probability of obtaining $b$ in the 2-outcome projective measurement is very close to $1$, then the postmeasurement state is very close to the initial state $\ket{\psi}$, which means that this state can be ``reused'' in later computation. In this paper, we will not need a formal statement of this second property. We will use instead (as a black box) the quantum union bound from Theorem \ref{thm:Gao15}. 



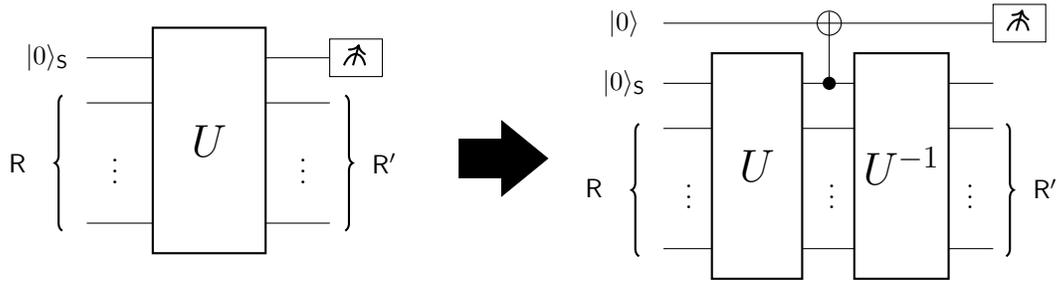
\begin{figure}[ht!]
    \begin{minipage}{.2\textwidth}
    \centering
    \begin{tikzpicture}[scale=0.4,rectnode/.style={thick,shape=rectangle,draw=black,fill=white,minimum height=30mm, minimum width=15mm},rectnode2/.style={shape=rectangle,draw=black,minimum height=5mm, minimum width=7mm},roundnode/.style={circle, draw=black, fill=black, minimum height=1mm}]
        \newcommand\XA{9.6}
        \newcommand\YA{5}
            
        \node[] (ancilla) at (-1.5,1.7*\YA) {};
        \node[] (R1) at (-1.5,1.4*\YA) {};
        \node[] (R4) at (-1.5,0.6*\YA) {};
        \node[] (RR1) at (0.75*\XA,1.4*\YA) {};
        \node[] (RR4) at (0.75*\XA,0.6*\YA) {};

        \draw (R1) -- (RR1);
        \draw (R4) -- (RR4);

        \node[rectnode2] (mes) at (0.6*\XA+2,1.7*\YA) {};
        \draw (ancilla) -- (mes);
        \draw[->, thick] (0.6*\XA+1.9,1.65*\YA) -- (0.6*\XA+2.1,1.79*\YA);
        \draw[thick] (0.6*\XA+1.6,1.65*\YA) .. controls  (0.6*\XA+1.8,1.75*\YA) and  (0.6*\XA+2.2,1.75*\YA) .. (0.6*\XA+2.3,1.65*\YA);

         \node[rectnode] (U) at (0.3*\XA,1.15*\YA) {\LARGE $U$};
         
         \node[draw=none,fill=none] at (-2.5,1.7*\YA) {$\ket{0}_{\mathsf{S}}$};
         \draw[thick,decorate,decoration={brace,amplitude=1mm}] (-2,-1+0.75*\YA) -- (-2,3.5+0.75*\YA);
         \node[draw=none,fill=none] at (-3.5,1*\YA) {$\sf{R}$};
         \node[draw=none,fill=none] at (6.00,1*\YA) {$\vdots$};
         \node[draw=none,fill=none] at (-0.25,1*\YA) {$\vdots$};
         \draw[thick,decorate,decoration={brace,amplitude=1mm}] (7.4,3.5+0.75*\YA) -- (7.4,-1+0.75*\YA);
         \node[draw=none,fill=none] at (8.7,1*\YA) {$\sf{R'}$};
    \end{tikzpicture}
    \end{minipage}
    \hspace{28mm}
    \begin{minipage}{.2\textwidth}
        \centering
        \begin{tikzpicture}[scale=0.4,rectnode/.style={thick,shape=rectangle,draw=black,fill=white,minimum height=30mm, minimum width=12mm},rectnode2/.style={shape=rectangle,draw=black,minimum height=5mm, minimum width=7mm},roundnode/.style={circle, draw=black, fill=black, minimum height=1mm}]
            \newcommand\XA{9.6}
            \newcommand\YA{5}
            \draw[black, fill=black, -{Triangle[width = 28pt, length = 18pt]}, line width = 15pt] (-8.0, 6.0) -- (-5.0, 6.0);
            \node[] (R0) at (-1.5,1.7*\YA) {};
            \node[] (RR0) at (1.05*\XA,1.7*\YA) {};
            \node[] (R1) at (-1.5,1.4*\YA) {};
            \node[] (R4) at (-1.5,0.6*\YA) {};
            \node[] (RR1) at (1.05*\XA,1.4*\YA) {};
            \node[] (RR4) at (1.05*\XA,0.6*\YA) {};
            \node[] (ancilla2) at (-1.5,2.1*\YA) {};
            
            \draw (R0) -- (RR0);
            \draw (R1) -- (RR1);
            \draw (R4) -- (RR4);
    
            \node[rectnode2] (mes) at (0.9*\XA+2,2.1*\YA) {};
            \draw (ancilla2) -- (mes);
            
            \draw[->, thick] (0.9*\XA+1.9,2.05*\YA) -- (0.9*\XA+2.1,2.19*\YA);
            \draw[thick] (0.9*\XA+1.6,2.05*\YA) .. controls  (0.9*\XA+1.8,2.15*\YA) and  (0.9*\XA+2.2,2.15*\YA) .. (0.9*\XA+2.3,2.05*\YA);
    
             \node[rectnode] (U) at (0.2*\XA,1.15*\YA) {\LARGE $U$};
             \node[rectnode] (U) at (0.7*\XA,1.15*\YA) {\LARGE $U^{-1}$};
             
             \node[draw=none,fill=none] at (-2.5,1.7*\YA) {$\ket{0}_{\mathsf{S}}$};
             \node[draw=none,fill=none] at (-2.5,2.1*\YA) {$\ket{0}$};
             \draw[thick,decorate,decoration={brace,amplitude=1mm}] (-2,-1+0.75*\YA) -- (-2,3.4+0.75*\YA);
             \draw[thick,decorate,decoration={brace,amplitude=1mm}] (10.2,3.4+0.75*\YA) -- (10.2,-1+0.75*\YA);
             \node[draw=none,fill=none] at (-3.5,1*\YA) {$\sf{R}$};
             \node[draw=none,fill=none] at (11.5,1*\YA) {$\sf{R'}$};
             \node[draw=none,fill=none] at (4.40,1*\YA) {$\vdots$};
             \node[draw=none,fill=none] at (-0.35,1*\YA) {$\vdots$};
             \node[draw=none,fill=none] at (9,1*\YA) {$\vdots$};

             \filldraw[black](0.45*\XA,1.7*\YA) circle (2mm);
             \draw(0.45*\XA,2.1*\YA) circle (4mm);
             \draw (0.45*\XA,1.7*\YA) -- (0.45*\XA,2.1*\YA+0.38);
     
        \end{tikzpicture}
        \end{minipage}
    \caption{Conversion from a quantum circuit implementing a (non-projective) 2-outcome measurement associated with the unitary $U$ (left) to a quantum circuit implementing a 2-outcome projective measurement (right). Register $\sf{R}$ stores the initial state and Register $\sf{R}'$ stores the postmeasurement state. In the right picture, the 2-qubit gate between $U$ and $U^{-1}$ represents the CNOT gate, where the $X$ gate (also called NOT gate) is applied on the $\oplus$-part conditioned on the content of the black-circle part being $1$.}\label{fig:PM}
    \end{figure}

    \clearpage
\section{Description of the SWAP test}\label{Appendix:SWAP}
In this appendix we give the technical description of the SWAP test presented in Section \ref{sec:prelim}: we describe the test in Figure~\ref{fig:SWAPprot} and give the corresponding quantum circuit in Figure \ref{fig:SWAP}.

\begin{figure}[ht!]
    \begin{center}
        \fbox{
         \begin{minipage}{13 cm} \vspace{2mm}
        
         \noindent{\bf SWAP test}\\\vspace{-3mm}
        
        \noindent\hspace{3mm} Input: two quantum states 
        in Registers ${\sf R}_1$ and ${\sf R}_2$, respectively  
    
        \vspace{4mm}
    
        \noindent\hspace{3mm}
        1. Introduce a 1-register $\mathsf{S}$ initialized to $\ket{0}_{\mathsf{S}}$.
        \vspace{2mm}
    
        \noindent\hspace{3mm}
        2. Apply the Hadamard gate $H$ on ${\sf S}$.
        \vspace{2mm}
        
        \noindent\hspace{3mm}
        3. (Controlled SWAP) If the content of ${\sf S}$ is $1$, swap ${\sf R}_1$ and ${\sf R}_2$.
        \vspace{2mm}
      
        \noindent\hspace{3mm} 
        4. Apply the Hadamard gate $H$ and then the $X$ gate on ${\sf S}$.
        \vspace{2mm}
    
        \noindent\hspace{3mm} 
        5. Measure Register ${\sf S}$ in the computational basis and output the outcome.
        \vspace{2mm}
         \end{minipage}
        }
    \end{center}\vspace{-2mm}
    \caption{Description of the SWAP test.}\label{fig:SWAPprot}
    \end{figure}

    \begin{figure}[ht!]
        \centering
        \begin{tikzpicture}[scale=0.5,rectnode/.style={thick,shape=rectangle,draw=black,fill=white,minimum height=30mm, minimum width=12mm},rectnode2/.style={shape=rectangle,draw=black,minimum height=5mm, minimum width=7mm},roundnode/.style={circle, draw=black, fill=black, minimum height=1mm},rectnode3/.style={dashed,shape=rectangle,draw=black,minimum height=45mm, minimum width=72mm}]
            \newcommand\XA{9.6}
            \newcommand\YA{6}
                
            \node[rectnode2] (H1) at (1.3,1.7*\YA) {$H$};
            \node[rectnode2] (H2) at (12.2,1.7*\YA) {$H$};
            \node[rectnode2] (X) at (14.2,1.7*\YA) {$X$};
            
            \node[] (ancilla) at (-1.5,1.7*\YA) {};
            \node[] (R1) at (-1.5,1.4*\YA) {};
            \node[] (R2) at (-1.5,1.1*\YA) {};
            \node[] (R3) at (-1.5,0.9*\YA) {};
            \node[] (R4) at (-1.5,0.6*\YA) {};
            \node[] (RR1) at (1.77*\XA,1.4*\YA) {};
            \node[] (RR2) at (1.77*\XA,1.1*\YA) {};
            \node[] (RR3) at (1.77*\XA,0.9*\YA) {};
            \node[] (RR4) at (1.77*\XA,0.6*\YA) {};
            \node[] (ancilla) at (-1.5,1.7*\YA) {};
            \node[rectnode2] (mes) at (1.6*\XA+2,1.7*\YA) {};
    
            \draw (R1) -- (RR1);
            \draw (R2) -- (RR2);
            \draw (R3) -- (RR3);
            \draw (R4) -- (RR4);
            \draw (ancilla) -- (H1) -- (H2) -- (X) -- (mes);
     
            \draw[->, thick] (1.6*\XA+1.9,1.65*\YA) -- (1.6*\XA+2.1,1.77*\YA);
            \draw[thick] (1.6*\XA+1.6,1.65*\YA) .. controls  (1.6*\XA+1.8,1.75*\YA) and  (1.6*\XA+2.2,1.75*\YA) .. (1.6*\XA+2.4,1.65*\YA);
             \node[rectnode] (SWAP) at (0.7*\XA,1*\YA) {SWAP};
             \draw (0.7*\XA,1.7*\YA) -- (SWAP);
             \filldraw[black](0.7*\XA,1.7*\YA) circle (2mm);

             \node[draw=none,fill=none] at (-2.5,1.7*\YA) {$\ket{0}_{\sf{S}}$};
             \draw[thick,decorate,decoration={brace,amplitude=1mm}] (-2,-1.2+0.75*\YA) -- (-2,1.2+0.75*\YA);
             \node[draw=none,fill=none] at (-3.5,1.25*\YA) {$\sf{R}_1$};
             \node[draw=none,fill=none] at (12.25,1.27*\YA) {$\vdots$};
             \node[draw=none,fill=none] at (1.25,1.27*\YA) {$\vdots$};
             
             \draw[thick,decorate,decoration={brace,amplitude=1mm}] (-2,-1.2+1.25*\YA) -- (-2,1.2+1.25*\YA);
             \node[draw=none,fill=none] at (-3.5,0.75*\YA) {$\sf{R}_2$}; 
             \node[draw=none,fill=none] at (1.25,0.77*\YA) {$\vdots$};
             \node[draw=none,fill=none] at (12.25,0.77*\YA) {$\vdots$};
     
        \end{tikzpicture}
        \caption{Quantum circuit for the SWAP test.}\label{fig:SWAP}
        \end{figure}
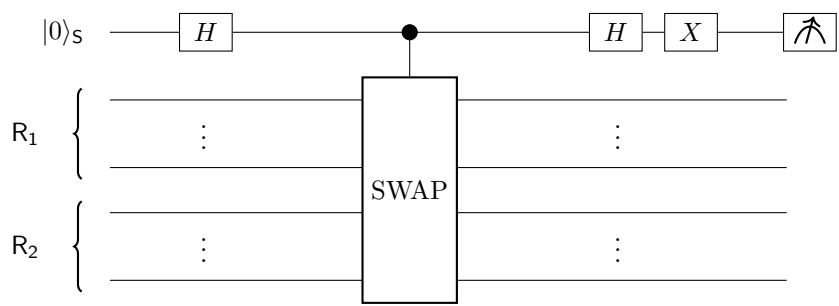

\end{document}